\documentclass[12pt,draftcls,onecolumn]{IEEEtran}

\usepackage{amsfonts}
\usepackage{cite} 
\usepackage{amsmath,amsthm} 
\usepackage{amssymb}  
\usepackage{graphicx}
\usepackage{tikz}
\usepackage{epsfig}
\usepackage{epstopdf}
\usepackage{cite}
\usepackage[makeroom]{cancel}
\usepackage{mathtools}
\usepackage{etoolbox}
\usepackage{lipsum}
\usepackage{caption}
\usepackage{subcaption}
\usepackage[colorlinks=true]{hyperref}
\usepackage{verbatim}
\usepackage{setspace}

\usepackage{algorithm,algorithmic}
\newenvironment{varalgorithm}[1]
  {\algorithm\renewcommand{\thealgorithm}{#1}}
  {\endalgorithm}

\let\bbordermatrix\bordermatrix
\patchcmd{\bbordermatrix}{8.75}{4.75}{}{}
\patchcmd{\bbordermatrix}{\left(}{\left[}{}{}
\patchcmd{\bbordermatrix}{\right)}{\right]}{}{}

\hypersetup{
     colorlinks = true,
     linkcolor = [rgb]{0,0,1},
     anchorcolor = [rgb]{0,0,1},
     citecolor = [rgb]{0.9,0.5,0},
     filecolor = [rgb]{0,0,1},
     pagecolor = [rgb]{1,1,0},
     urlcolor = [rgb]{0.9,0.5,0},
     bookmarks,
        bookmarksopen = true,
        bookmarksnumbered = true,
        breaklinks = true,
        linktocpage,
        pagebackref,
        colorlinks = true,
        linkcolor = [rgb]{0.2,0.6,0.2},
        urlcolor  = [rgb]{0,0,1},
        citecolor = [rgb]{0.9,0.5,0},
        anchorcolor = [rgb]{0.2,0.6,0.2},
        hyperindex = true,
        hyperfigures
}

\def\IID{\mathop{\mathrm{IID}}}  
\def\diag{\mathop{\mathrm{diag}}}   
\def\trace{\mathop{\mathrm{trace}}} 
\def\op{\mathop{\mathrm{op}}}       
\def\rhs{\mathop{\mathrm{RHS}}}     
\def\mse{\mathop{\mathrm{MSE}}}     
\def\mmse{\mathop{\mathrm{MMSE}}}     
\def\nrdf{\mathop{\mathrm{NRDF}}}   
\def\rdf{\mathop{\mathrm{RDF}}}   
\def\rvs{\mathop{\mathrm{RVs}}}     
\def\rv{\mathop{\mathrm{RV}}}       
\def\pmf{\mathop{\mathrm{PMF}}}     
\def\mc{\mathop{\mathrm{MC}}}       
\def\ssum{\mathop{\mathrm{sum}}}       
\def\lqg{\mathop{\mathrm{LQG}}}       
\def\ecdq{\mathop{\mathrm{ECDQ}}}       
\def\dpcm{\mathop{\mathrm{DPCM}}}       

\newcommand{\T}{^{\mbox{\tiny T}}}

\newtheorem{theorem}{Theorem}
\newtheorem{lemma}{Lemma}
\newtheorem{definition}{Definition}
\newtheorem{corollary}{Corollary}

\newtheorem{remark}{Remark}

\newcommand{\sr}{\stackrel}

\newcommand{\be}{\begin{equation}}
\newcommand{\ee}{\end{equation}}
\newcommand{\bea}{\begin{eqnarray}}
\newcommand{\eea}{\end{eqnarray}}
\newcommand{\bes}{\begin{eqnarray*}}
\newcommand{\ees}{\end{eqnarray*}}

\newcommand{\bfi}{\begin{figure}}
\newcommand{\bfit}{\begin{figure}[t]}
\newcommand{\bfib}{\begin{figure}[b]}
\newcommand{\bfih}{\begin{figure}[h]}
\newcommand{\bfip}{\begin{figure}[p]}
\newcommand{\efi}{\end{figure}}

\newcommand{\bi}{\begin{itemize}}
\newcommand{\ei}{\end{itemize}}
\newcommand{\ben}{\begin{enumerate}}
\newcommand{\een}{\end{enumerate}}

\newcommand*{\QEDA}{\hfill\ensuremath{\blacksquare}}%

\begin{document}

\sloppy
\title{Sequential Source Coding for Stochastic Systems Subject to Finite Rate Constraints} 

\author{{\IEEEauthorblockN{Photios A. Stavrou, Mikael Skoglund and Takashi Tanaka
\thanks{P. A. Stavrou and M. Skoglund received funding by the KAW Foundation and the Swedish Foundation for Strategic Research.}
\thanks{\IEEEauthorblockA{Photios A. Stavrou and Mikael Skoglund are with the Division of Information Science and Engineering, KTH Royal Institute of Technology, Stockholm, Sweden}, {\it emails:\{fstavrou,skoglund\}@kth.se}}
\thanks{\IEEEauthorblockA{T. Tanaka is with the Department of Aerospace Engineering and Engineering Mechanics, University of Austin, TX, USA}, {\it email: ttanaka@utexas.edu}}
 }
 }}
%
%
\maketitle

%
%
%
%
\begin{abstract}
In this paper, we revisit the sequential source coding framework to analyze fundamental performance limitations of discrete-time stochastic control systems subject to feedback data-rate constraints in finite-time horizon. The basis of our results is a new characterization of the lower bound on the minimum total-rate achieved by sequential codes subject to a total (across time) distortion constraint and a computational algorithm that allocates optimally the rate-distortion for any fixed finite-time horizon. This characterization facilitates the derivation of {\it analytical}, {\it non-asymptotic},  and {\it finite-dimensional} lower and upper bounds in two control-related scenarios. (a) A parallel time-varying Gauss-Markov process with identically distributed spatial components that is quantized and transmitted through a noiseless channel to a minimum mean-squared error (MMSE) decoder. (b) A time-varying quantized LQG closed-loop control system, with identically distributed spatial components and with a random data-rate allocation. Our non-asymptotic lower bound on the quantized LQG control problem, reveals the absolute minimum data-rates for (mean square) stability of our time-varying plant for any fixed finite time horizon. We supplement our framework with illustrative simulation experiments.
\end{abstract}


%
%
%
%
\section{Introduction}\label{sec:introduction}

\par One of the fundamental characteristics of networked control systems (\text{NCSs}) \cite{zhang:2016} is the existence of an imperfect communication network between computational and physical entities. In such setups, an analytical framework to assess impacts of communication and data-rate limitations on the control performance is strongly required.  
\par In this paper, we adopt information-theoretic tools to analyze these requirements. Specifically, we consider {\it sequential coding of correlated sources} {initially introduced by} \cite{viswanathan:2000} (see also \cite{tatikonda:2000}) (see Fig. \ref{fig:sequential_coding}), which is a generalization of the successive refinement source coding problem \cite{koshelev:1980,equitz:1991,rimoldi:1994}. In successive refinement, source coding is performed in (time) stages where one first describes the given source within a few bits of information and, then, tries to ``refine'' the description of the same source (at the subsequent stages) when more information is available. Sequential coding differs from successive refinement in that at the second stage, encoding involves describing a correlated (in time) source as opposed to improving the description of the same source. To accomplish this task, sequential coding encompasses a spatio-temporal coding method. 
\begin{figure*}[htp]
\centering
\includegraphics[width=0.7\textwidth]{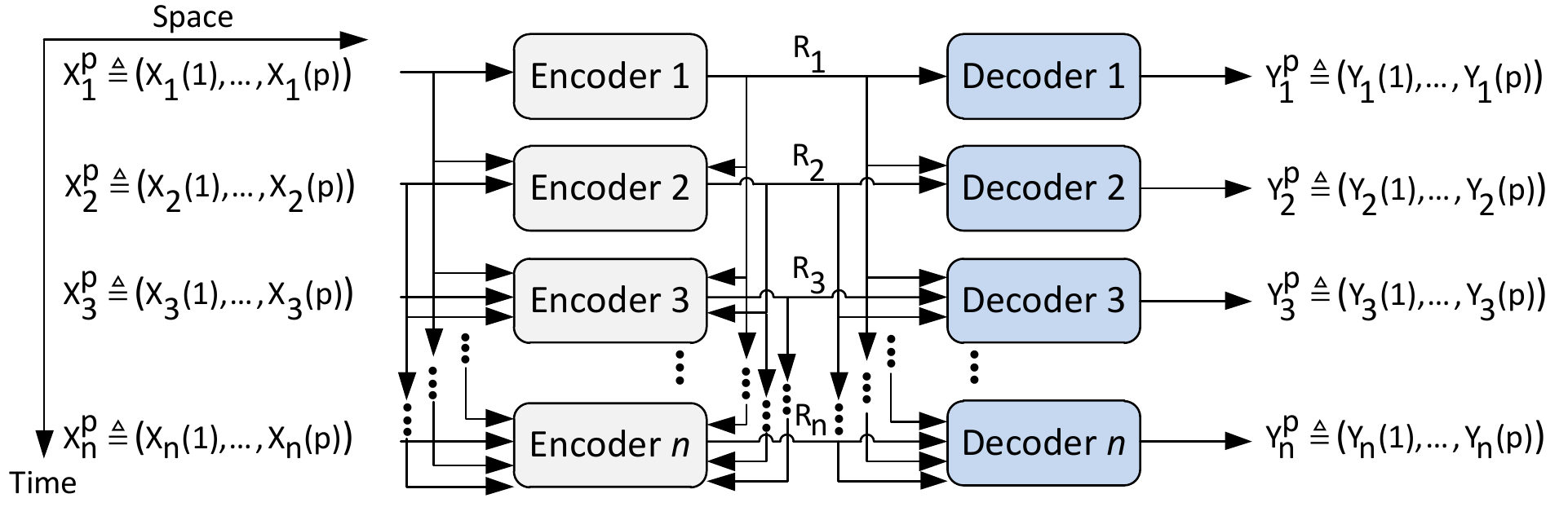}
\caption{Sequential coding of correlated sources.} 
\label{fig:sequential_coding}
\end{figure*}
In addition, sequential coding is a temporally zero-delay coding paradigm since both {\it encoding} and {\it decoding} must occur in real-time. The resulting zero-delay coding approach should not be confused with other existing works on zero-delay coding, see, e.g., \cite{witsenhausen:1979,gaarder-slepian:1982,teneketzis:2006,linder:2014,wood:2017,stavrou:2018stsp}, because it relies on the use of a spatio-temporal coding approach (see Fig. \ref{fig:sequential_coding}) whereas the aforementioned papers rely solely on temporal coding approaches. 

\subsection{{Literature review on sequential source coding}}\label{subsec:literature_review}

{In what follows, we provide a detailed literature review on sequential source coding. However, in order to shed more light on the historical route of this coding paradigm, we distinguish the work of \cite{viswanathan:2000} (see also \cite{ma-ishwar:2011,yang:2011}) with the work of \cite{tatikonda:2000} because although their results complement each other, their underlining motivation has been different. Indeed, \cite{viswanathan:2000} initiated this coding approach targeting video coding applications, whereas \cite{tatikonda:2000} aimed to develop a framework for delay-constrained systems and to study the communication theory in classical closed-loop control setups.}

\paragraph*{{Sequential coding via \cite{viswanathan:2000}}} The authors of \cite{viswanathan:2000} characterized the minimum achievable rate-distortion region for two temporally correlated random variables with each being a vector of spatially independent and identically distributed ($\IID$) processes (also called ``frames'' or spatial vectors), subject to a coupled average distortion criterion.  
The last decade, sequential coding approach of \cite{viswanathan:2000}  was further studied in \cite{ma-ishwar:2011,yang:2011,yang:2014}. In \cite{ma-ishwar:2011}, the authors used an extension of the framework of \cite{viswanathan:2000} to three time instants subject to a per-time distortion constraint to investigate the effect of sequential coding when possible coding delays occur within a multi-input multi-output system. Around the same time, \cite{yang:2011} generalized the framework of \cite{viswanathan:2000} to a finite number of time instants. Compared to \cite{viswanathan:2000} and \cite{ma-ishwar:2011}, their spatio-temporal source process is correlated over time whereas each frame is spatially jointly stationary and totally ergodic subject to a per-time average distortion criterion. More recently, the same authors in \cite{yang:2014} drew connections between sequential causal coding and {\it predictive sequential causal coding}, that is, for (first-order) Markov sources subject to a single-letter fidelity constraint, sequential causal coding and sequential predictive coding coincide. For three time instants of an $\IID$ vector source containing jointly Gaussian correlated processes (not necessarily Markov) an explicit expression of the minimum achievable sum-rate for a per-time mean-squared error ($\mse$) distortion is obtained in \cite{torbatian:2012}. Inspired by the framework of \cite{viswanathan:2000,ma-ishwar:2011}, Khina {\it et al.} in \cite{khina:2018}  derived fundamental performance limitations in control-related applications. In their work, they considered a multi-track system that tracks several parallel time-varying Gauss-Markov processes with $\IID$ spatial components conveyed over a single shared wireless communication link (possibly prone to packet drops) to a minimum mean-squared error ($\mmse$) decoder. In their Gauss-Markov multi-tracking scenario, they provided lower and upper bounds in finite-time and in the per unit time asymptotic limit for the distortion-rate region of time-varying Gauss-Markov sources subject to a mean-squared error ($\mse$) distortion constraint. Their lower bound is characterized by a forward in time distortion allocation algorithm operating with {\it given data-rates at each time instant for a finite time horizon} whereas their upper bound is obtained by means of a differential pulse-code modulation (\text{DPCM}) scheme using entropy coded dithered quantization (\text{ECDQ}) using one dimensional lattice constrained by {\it data rates averaged across time} (for details on this coding scheme, see, e.g., \cite{farvardin:1985,zamir:2014}). Subsequently, they used these bounds in a scalar-valued quantized linear quadratic Gaussian (\text{LQG}) closed-loop control problem to find similar bounds on the minimum cost of control.
\paragraph*{{Sequential coding via \cite{tatikonda:2000}}} {A similar framework to \cite{viswanathan:2000} was independently introduced and developed by Tatikonda in \cite[Chapter 5]{tatikonda:2000} (see also \cite{borkar:2001}) in the context of delay-constrained and control-related applications. Tatikonda in \cite{tatikonda:2000}, introduced an information theoretic quantity called {\it sequential rate distortion function ($\rdf$)} that is attributed to the works of Gorbunov and Pinsker in \cite{gorbunov-pinsker1972b,gorbunov-pinsker1972a}. Using the sequential $\rdf$, Tatikonda {\it et al.} in \cite{tatikonda:2004} studied the performance analysis and synthesis of a multidimensional fully observable time-invariant Gaussian closed-loop control system when a communication link exists between a stochastic linear plant and a controller whereas the performance criterion is the classical linear quadratic cost. The use of sequential $\rdf$ (also termed nonanticipative or causal $\rdf$ in the literature) in filtering applications is stressed in \cite{charalambous:2014,tanaka:2017tac,stavrou:2018siam}. Analytical expressions of lower and upper bounds for the setup of \cite{tatikonda:2004} including the cases where a linear fully observable time-invariant plant is driven by $\IID$ non-Gaussian noise processes or when the system is modeled by time-invariant partially observable Gaussian processes are derived in \cite{kostina:2019}. Tanaka {\it et al.} in \cite{tanaka:2016,tanaka:2018} studied the performance analysis and synthesis of a linear fully observable and partially observable Gaussian closed loop control problem when the performance criterion is the linear quadratic cost. Moreover, they showed that one can derive lower bounds in finite time and in the per unit time asymptotic limit by casting the problems as semidefinite representable and thus numerically computable by known solvers. An achievability bound on the asymptotic limit using a \text{DPCM}-based $\ecdq$ scheme that uses one dimensional quantizer at each dimension was also proposed. Lower and upper bounds for a general closed-loop control system subject to {\it  asymptotically average total data-rate constraints} across the time are also investigated in \cite{silva:2011,silva:2016}. The lower bounds are obtained using sequential coding and directed information \cite{massey:1990} whereas the upper bounds are obtained via a sequential $\ecdq$ scheme using scalar quantizers.}

\subsection{Contributions}\label{subsec:contributions}

{In this paper, we first revisit the sequential coding framework developed by \cite{viswanathan:2000,tatikonda:2000,ma-ishwar:2011,yang:2011} to obtain the following main results.
\begin{itemize}
\item[{\bf (1)}]  {Analytical} non-asymptotic {and finite-dimensional} lower and upper bounds on the minimum achievable total-rates (per-dimension) for a multi-track communication scenario similar to the one considered in \cite{khina:2018}. However, compared to \cite{khina:2018}, who derived distortion-rate bounds via forward recursions with {\it given data rates across a finite time horizon}, here we derive a lower bound subject to a dynamic reverse-waterfilling algorithm in which for a given distortion threshold $D>0$ we optimally assign the data-rates and the $\mse$ distortions at each time instant for a finite time horizon (Theorem \ref{theorem:fundam_lim_mmse}). We also implement our algorithm in Algorithm \ref{algo1}. Our lower bound is the basis to derive our upper bound on the minimum achievable total-rates (per dimension) using a sequential $\dpcm$-based $\ecdq$ scheme that is constrained by total-rates for a finite time horizon. For the specific rate constraint we use a dynamic reverse-waterfilling algorithm obtained from our lower bound to allocate the rate and the $\mse$ distortion at each time instant for the whole finite time horizon. This rate constraint is the fundamental difference compared to similar upper bounds derived in \cite[Theorem 6]{khina:2018} and \cite[Corollary 5.2]{silva:2011} (see also \cite{silva:2016,stavrou:2018stsp}) that {\it restrict their transmit rates to have fixed rates that are averaged across the time horizon or that are asymptotically averaged across the time}.
\item[{\bf (2)}] We obtain {analogous bounds to {\bf (1)}} on the minimum achievable total (across time) cost-rate function of control (per-dimension) for a \text{NCS} with time-varying quantized {$\lqg$} closed-loops {operating with data-rate obtained subject to a solution of a reverse-waterfilling algorithm} (Theorems \ref{theorem:fund_limt_lqg}, \ref{theorem:achievability_lqg}). 
\end{itemize}
}
{\noindent{\it Discussion of the contributions and additional results.} The non-asymptotic lower bound in {\bf (1)} is obtained because for parallel processes all involved matrices in the characterization of the corresponding optimization problem {\it commute by pairs} \cite[p. 5]{harville:1997} thus they are {\it simultaneously diagonalizable} by an orthogonal matrix \cite[Theorem 21.13.1]{harville:1997} and the resulting optimization problem simplifies to one that resembles scalar-valued processes. The upper bound in {\bf (1)} is obtained because we are able to employ a lattice quantizer \cite{zamir:2014} using a quantization scheme with existing performance guarantees such as the $\dpcm$-based $\ecdq$ scheme and using existing approximations from quantization theory for high-dimensional but possibly finite-dimensional quantizers with a $\mse$ performance criterion (see, e.g., \cite{conway-sloane1999}). {The non-asymptotic bounds derived in {\bf (2)} are obtained using the so-called ``weak separation principle'' of quantized $\lqg$ control (for details, see \S\ref{sec:ncs}) and well-known inequalities that are used in information theory. Interestingly, our lower bound in {\bf (2)} also reveals the minimum allowable data rates on the cost-rate (or rate-cost) function in control at each time instant to ensure (mean square) stability of the plant (see e.g., \cite{nair:2004} for the definition) for the specific \text{NCS} (Remark \ref{remark:technical_remarks_lower_bound})}. Finally, for every bound in this paper, we derive the corresponding bounds in the infinite time horizon recovering several known results in the literature (see Corollaries \ref{corollary:steady_state_rev_water}-\ref{corollary:steady_state_lqg_upper}).} 
\par {This paper is organized as follows. In \S\ref{sec:preliminaries} we give an overview of known results on sequential coding. In \S\ref{sec:examples} we derive non-asymptotic bounds   and their corresponding per unit time asymptotic limits for a quantized state estimation problem. In \S\ref{sec:ncs}, we use the results of \S\ref{sec:examples} and the weak separation principle to derive non-asymptotic bounds and their corresponding per unit time asymptotic limits for a quantized $\lqg$ closed-loop control problem. In \S\ref{sec:discussion} we discuss several open questions that can be answered based on this work and draw conclusions in \S\ref{sec:conclusions}. 
}
\paragraph*{\bf Notation} $\mathbb{R}$ is the set of real numbers, $\mathbb{N}_{1}$ is the set of positive integers, and $\mathbb{N}_1^n\triangleq\{1,\ldots,n\}$,~ $n\in\mathbb{N}_1$, respectively.  Let $\mathbb{X}$ be a finite-dimensional Euclidean space, and ${\cal B}(\mathbb{X})$ be the Borel $\sigma$-algebra on $\mathbb{X}$. A random variable ($\rv$) $X$ defined on some probability space ($\Omega, {\cal F}, {\bf P}$)
is a map $X : \Omega\mapsto \mathbb{X}$. The probability distribution of a $\rv$ $X$ with realization $X=x$ on $\mathbb{X}$ is denoted by ${\bf P}_X\equiv{p}(x)$. The conditional distribution of a $\rv$ $Y$ with realization  $Y=y$, given $X=x$ is denoted by ${\bf Q}_{Y|X}\equiv {q}(y|x)$. We denote the sequence of one-sided $\rvs$ by $X_{t,j} \triangleq (X_{t}, X_{t+1}, \ldots,X_j),~{t}\leq{j},~(t,j)\in {\mathbb N}_1\times\mathbb{N}_1$, and their values by $x_{t,j} \in  {\mathbb X}_{t,j} \triangleq \times_{k={{t}}}^j {\mathbb X}_k$. We denote the sequence of ordered $\rvs$ with ``$i^{\text{th}}$'' spatial components by $X_{t,j}^i$, so that $X_{t,j}^i$ is a vector of dimension ``$i$'', and their values by $x_{t,j}^i \in  {\mathbb X}_{t,j}^i \triangleq \times_{k={{t}}}^j {\mathbb X}^i_k$, where ${\mathbb X}^i_k\triangleq\left(\mathbb{X}_k(1),\ldots\mathbb{X}_k(i)\right)$. The notation ${X}\leftrightarrow{Y}\leftrightarrow{Z}$ denotes a Markov Chain ($\mc$) which means that $p(x|y,z)=p(x|y)$. We denote the diagonal of a square matrix by $\diag(\cdot)$ and the $p\times{p}$ identity matrix by $I_p$. If $A\in\mathbb{R}^{p{\times}{p}}$, we denote by $A\succeq{0}$ (resp., $A\succ{0}$) a positive semidefinite matrix (resp., positive definite matrix). We denote the determinant and trace of some matrix $A\in\mathbb{R}^{p\times{p}}$ by $|A|$ and $\trace(A)$, respectively. {We denote by $h(x)$ (resp. $h(x|y)$) the differential entropy of a distribution $p(x)$ (resp. $p(x|y)$). We denote ${\cal D}(P||Q)$ the relative entropy of probability distributions $P$ and $Q$. We denote by ${\bf E}\{\cdot\}$ the expectation operator and $||\cdot||_2$ the Euclidean norm.} Unless otherwise stated, when we say ``total'' distortion, ``total-rate'' or ``total-cost'' we mean with respect to time. Similarly, by referring to ``average total'' we mean normalized over the total finite time horizon. 

%
%
%
%
\section{Known Results on Sequential Coding}\label{sec:preliminaries}

\par {In this section, we give an overview of the sequential causal coding introduced and analyzed independently by \cite[Chapter 5]{tatikonda:2000} and \cite{viswanathan:2000,ma-ishwar:2011,yang:2011}. We merge both frameworks because some results obtained in \cite{ma-ishwar:2011,yang:2011} complement the results of \cite[Chapter 5]{tatikonda:2000} and vice versa.} 
\par In the following analysis, we will consider processes for a fixed time-span $t\in\mathbb{N}^n_1$, i.e., ($X_1,\ldots,X_{n}$). Following \cite{ma-ishwar:2011,yang:2011}, we assume that the sequences of $\rvs$ are defined on alphabet spaces with finite cardinality. 
Nevertheless, these can be extended following for instance the techniques employed in \cite{oohama:1998} to continuous alphabet spaces as well (i.e., Gaussian processes) with $\mse$ distortion constraints. 



First, we use some definitions (with slight modifications to ease the readability of the paper) from \cite[\S{II}]{ma-ishwar:2011} and \cite[\S{I}]{yang:2011}.

\begin{definition}(Sequential causal coding)\label{def:sequential_codes_iid}
A spatial order $p$ sequential causal code ${\cal C}_p$ for the (joint) vector source ($X_1^p,X_2^p,\ldots,X_n^p$) is formally defined by a sequence of encoder and decoder pairs ($f^{(p)}_1,g^{(p)}_1$),$\ldots$,($f^{(p)}_n,g^{(p)}_n$) such that
\begin{align}
\begin{split}\label{seq:coding_functions}
f_t^{(p)}&:~\mathbb{X}^p_{1,t}\times\underbrace{\{0,1\}^*\times\ldots\times\{0,1\}^*}_{t-1~times}\longrightarrow\mathbb\{0,1\}^*\\
g_t^{(p)}&:~\underbrace{\{0,1\}^*\times\ldots\times\{0,1\}^*}_{t~times}\longrightarrow\mathbb{Y}^p_t,~t\in\mathbb{N}_1^n
\end{split},
\end{align}
where $\{0,1\}^*$ denotes the set of all binary sequences of finite length satisfying the property that at each time instant $t$ the range of $\{f_t:~t\in\mathbb{N}_1^n\}$ given any $t-1$ binary sequences is an instantaneous code. Moreover, the encoded and reconstructed sequences of $\{X_t^p:~t\in\mathbb{N}_1^n\}$ are given by $S_t=f_t(X^p_{1,t},S_{1,t-1})$, with $S_t\in\mathbb{S}_t\subset\{0,1\}^*$, and $Y_t^p=g_t(S_{1,t})$, respectively, with $|\mathbb{Y}_t|<\infty$. Moreover, the expected rate in bits per symbol at each time instant (normalized over the spatial components) is defined as
\begin{align}
r_t\triangleq\frac{{\bf E}|S_t|}{p},~t\in\mathbb{N}_1^n,\label{coding_rate:sequential}
\end{align}
where $|S_t|$ denotes the length of the binary sequence $S_t$. 
\end{definition}
\paragraph*{Distortion criterion} For each $t\in\mathbb{N}_1^n$, we consider a  total (in dimension) single-letter distortion criterion. This means that the distortion between $X^p_t$ and $Y^p_t$ is measured by a function $d_t:~\mathbb{X}^p_t\times\mathbb{Y}^p_t\longrightarrow[0,\infty)$ with maximum distortion $d_t^{\max}=\max_{x^p_t,y^p_t}d_t(x^p_t,y_t^p)<\infty$ such that 
\begin{align}
d_t(x^p_t,y^p_t)\triangleq\frac{1}{p}\sum_{i=1}^p{d}_t(x_t(i),y_t(i)).\label{distortion_function}
\end{align}
The per-time average distortion is defined as
\begin{align} 
{\bf E}\left\{d_t(X_t^p,Y_t^p)\right\}\triangleq\frac{1}{p}\sum_{i=1}^p{\bf E}\left\{{d}_t(X_t(i),Y_t(i))\right\}.\label{average_distortion}
\end{align}
We remark that {the following results are still valid even if} the distortion function \eqref{distortion_function} has dependency on previous reproductions $\{Y^p_{1,t-1}:~t\in\mathbb{N}_1^n\}$ (see, e.g., \cite{ma-ishwar:2011}). 
\begin{definition}(Achievability)\label{def:seq:operation_meaning}
 A rate-distortion tuple $(R_{1,n},D_{1,n})\triangleq(R_1,\ldots,R_n,D_1,\ldots,D_n)$ for any ``$n$'' is said to be {\it achievable} for a given sequential causal coding system if for all $\epsilon>0$, there exists a sequential code $\{(f^{(p)}_t,g^{(p)}_t):~t\in\mathbb{N}_1^n\}$ {such that there exists ${\cal P}$ for which
\begin{align}
\begin{split}\label{seq:achievability}
r_t&\leq{R_t}+\epsilon,\\
{\bf E}\left\{d_t(X_t^p,Y_t^p)\right\}&\leq{D_t}+\epsilon,~D_t\geq{0},~\forall{t}\in\mathbb{N}_1^n,
\end{split}
\end{align}
holds $\forall{p}\geq{\cal P}$}. Moreover, let the set of all achievable rate-distortion tuples $(R_{1,n},D_{1,n})$ be denoted by ${\cal R}^{*}$. Then, the minimum total-rate required to achieve the distortion tuple $(D_{1},~D_2,\ldots,D_n)$ is defined by the following optimization problem:
\begin{align}
{\cal R}^{\op}_{\ssum}(D_{1,n})\triangleq\inf_{(R_{1,n},D_{1,n})\in{\cal R}^{*}}\sum_{t=1}^nR_t.\label{operational_min_sum_rate}
\end{align}
\end{definition}

\paragraph*{Source model}  {The finite alphabet source randomly generates symbols $X^p_{1,n}=x^p_{1,n}\in\mathbb{X}^p_{1,n}$ according to the following temporally correlated joint probability mass function ($\pmf$) 
\begin{align}
{p}(x_{1,n}^p)\triangleq\otimes_{i=1}^p{p}(x_1(i),\ldots,x_{n}(i)),\label{source_correlated_1}
\end{align}
where the joint process $\{(X_1(i),\ldots,X_n(i))\}_{i=1}^p$ is identically distributed. This means that for each $i=1,\ldots,p$, the temporally correlated joint process $(X_1(i),\ldots,X_{n}(i))$ is independent of every other temporally correlated joint process $(X_1(j),\ldots,X_{n}(j))$, such that $i\neq{j}$. Furthermore, each temporally correlated joint process $(X_1(i),\ldots,X_{n}(i))$ is spatially identically distributed.} 

\paragraph*{Achievable rate-distortion regions and minimum achievable total-rate} Next, we characterize the achievable rate-distortion regions and the minimum achievable total-rate for the source model \eqref{source_correlated_1} with the distortion constraint \eqref{average_distortion}.

\par The following lemma is given in \cite[Theorem 5]{yang:2011}.  
\begin{lemma}(Achievable rate-distortion region)\label{lemma:coding_theorem_iid}
Consider the source model \eqref{source_correlated_1} with the average distortion of \eqref{average_distortion}. {Then, the {``spatially''} single-letter characterization of the rate-distortion region $(R_{1,n},~D_{1,n})$ is given by}:
\begin{align}
\begin{split}
{\cal R}^{\IID}&=\Bigg\{(R_{1,n},D_{1,n})\Bigg{|}\exists S_{1,n-1}, Y_{1,n},~\{g_t(\cdot)\}_{t=1}^n,\\
\text{s.t.}&\quad R_1\geq{I}(X_1;S_1),~\mbox{(initial time)}\\
&\quad R_t\geq{I}(X_{1,t};S_t|S_{1,t-1}),~t=2,\ldots,n-1,\\
&\quad R_n\geq{I}(X_{1,n};Y_n|S_{1,n-1}),~\mbox{(terminal time)},\\
&\quad D_t\geq{\bf E}\left\{d_t(X_t,Y_t)\right\},~t\in\mathbb{N}_1^n,\\
&\quad Y_1=g_1(S_1),~Y_{t}=g_t(S_{1,t}),~t=2,\ldots,n-1,\\
&\quad S_1\leftrightarrow(X_{1})\leftrightarrow{X}_{2,n},\\
&\quad S_t\leftrightarrow(X_{1,t},S_{1,t-1})\leftrightarrow{X}_{t+1,n},~t=2,\ldots,n-1\Bigg\},
\end{split}\label{single_letter_char_iid}
\end{align} 
where $\{S_{1,n-1},Y_{1,n}\}$ are the auxiliary (encoded) and reproduction $\rvs$, respectively, taking values in some finite alphabet spaces $\{\mathbb{S}_{1,n-1},\mathbb{Y}_{1,n}\}$, and $\{g_t(\cdot):~t\in\mathbb{N}_1^n\}$ are deterministic functions. 
\end{lemma} 
\begin{remark}(Comments on Lemma \ref{lemma:coding_theorem_iid})\label{remark:lemma1}
In the characterization of Lemma \ref{lemma:coding_theorem_iid}, the spatial index is excluded because the rate and distortion regions are normalized with the total number of spatial components. This point is also shown in \cite[Theorem 4]{yang:2011}. Following \cite{ma-ishwar:2011} or \cite{yang:2011}, Lemma \ref{lemma:coding_theorem_iid} gives a set ${\cal R}^{\IID}$ that is convex and closed  (this can be shown by trivially generalizing the time-sharing and continuity arguments of \cite[Appendix C2]{ma-ishwar:2011} to $n$ time-steps). This in turn means that ${\cal R}^*={\cal R^{\IID}}$ (see, e.g., \cite[Theorem 5]{yang:2011}). Thus, \eqref{operational_min_sum_rate} can be reformulated to the following optimization problem:
\begin{align}
{\cal R}_{\ssum}^{\IID,\op}(D_{1,n})\triangleq\min_{(R_{1,n},D_{1,n})\in{\cal R}^{\IID}}\sum_{t=1}^nR_t.\label{operational_min_sum_rate_new}
\end{align}
\end{remark}
\par {In what follows, we state a lemma that gives a lower bound on ${\cal R}_{\ssum}^{\IID,\op}(D_{1,n})$. The lemma is  stated without a proof as it is already derived in various papers, e.g., \cite[Theorem 5.3.1, Lemma 5.4.1]{tatikonda:2000}, \cite[Theorem 4.1]{silva:2011}, \cite[Corollary 1.1]{ma-ishwar:2011} (for $n=3$-time steps but can be trivially generalized to an arbitrary number of time-steps).}  
{
\begin{lemma}(Lower bound on \eqref{operational_min_sum_rate_new})\label{lemma:total_rate}
For $p$ sufficiently large, the following lower bound holds: 
\begin{align}
{\cal R}_{{\ssum}}^{\IID,\op}(D_{1,n})&\geq{\cal R}_{{\ssum}}^{\IID}(D_{1,n})\nonumber\\
&\triangleq\min_{\substack{{\bf E}\left\{d_t(X_t,Y_t)\right\}\leq{D}_t,~t\in\mathbb{N}_1^n\\Y_1\leftrightarrow{X_{1}}\leftrightarrow{X}_{2,n},\\~Y_t\leftrightarrow(X_{1,t},Y_{1,t-1})\leftrightarrow{X}_{t+1,n},~t=2,\ldots,n-1}}I(X_{1,n};Y_{1,n}),\label{eq:sumrate}
\end{align}
where $I(X_{1,n};Y_{1,n})=\sum_{t=1}^nI(X_{1,t};Y_t|Y_{1,t-1})$ is a variant of directed information \cite{marko1973,massey:1990} obtained by the conditional independence constraints imposed in the constraint set of \eqref{eq:sumrate}. 
\end{lemma}
}
\par {We note that the lower bound in Lemma \ref{lemma:total_rate} is often encountered in the literature by the name nonanticipatory $\epsilon-$entropy and sequential or nonanticipative $\rdf$.
}{
\begin{remark}(When do we achieve the lower bound in \eqref{eq:sumrate}?)\label{remark:ideal_quantization}
It should be noted that in \cite[Theorem 4]{yang:2011} it was shown via an algorithmic approach (see also \cite[Theorem 5]{yang:2011} for an equivalent proof via a direct and converse coding theorem) that Lemma \ref{lemma:total_rate} is achieved with equality if the number of $\IID$ spatial components tends to infinity, i.e., $p\longrightarrow\infty$, which also means that the optimal minimizer or ``test-channel'' at each time instant in \eqref{eq:sumrate}, corresponds precisely to the distribution generated by a sequential encoder, i.e., $S_t=Y_t$, for any $t\in\mathbb{N}_1^n$ (see also the derivation of \cite[Corollary 1.1]{ma-ishwar:2011}). In other words, the equality holds if the encoder (or quantizer for continuous alphabet sources) simulates exactly the corresponding ``test-channel''  distribution of \eqref{eq:sumrate}. This claim was also demonstrated via an application example for jointly Gaussian $\rvs$ and per-time $\mse$ distortion in \cite[Corollary 1.2]{ma-ishwar:2011} and also stated as a corollary referring to an ``ideal'' $\dpcm$-based $\mse$ quantizer in \cite[Corollary 1.3]{ma-ishwar:2011}. In general, however, for any $p<\infty$, the equality in \eqref{eq:sumrate} is not achievable.
\end{remark}
Next, we state the generalization of Lemma \ref{lemma:total_rate} when the constrained set is subject to an average total distortion constraint defined as $\frac{1}{n}\sum_{t=1}^n{\bf E}\left\{d_t(X_t,Y_t)\right\}\leq{D}$ with ${\bf E}\left\{d_t(X_t,Y_t)\right\}$ given in \eqref{average_distortion}.  This lemma was derived in \cite[Theorem 5.3.1, Lemma 5.4.1]{tatikonda:2000}. 
\begin{lemma}(Generalization of Lemma \ref{lemma:total_rate})\label{lemma:total_rate_aver_dist}
For $p$ sufficiently large, the following lower bound holds:
\begin{align}
\begin{split}
{\cal R}_{{\ssum}}^{\IID,\op}(D)&\geq{\cal R}_{{\ssum}}^{\IID}(D)\\
&=\min_{\substack{\frac{1}{n}\sum_{t=1}^n{\bf E}\left\{d_t(X_t,Y_t)\right\}\leq{D},~t\in\mathbb{N}_1^n\\Y_1\leftrightarrow(X_{1})\leftrightarrow{X}_{2,n},\\~Y_t\leftrightarrow(X_{1,t},Y_{1,t-1})\leftrightarrow{X}_{t+1,n},~t\in\mathbb{N}_2^{n-1}}}I(X_{1,n};Y_{1,n}),
\end{split}\label{eq:sumrate_aver_dist}
\end{align}
\end{lemma}
Clearly, one can use the same methodology applied in \cite[Theorems 4, 5]{yang:2011} to demonstrate that the lower bound in \eqref{eq:sumrate_aver_dist} is achieved once $p\longrightarrow\infty$ (see the discussion in Remark \ref{remark:ideal_quantization}). However, we once again point out that in general, \eqref{eq:sumrate_aver_dist} is a lower bound on the minimum achievable rates achieved by causal sequential codes.}

\paragraph*{{Information structures}} {Next, we state a few well-known structural results related to the bounds in Lemmas \ref{lemma:total_rate}, \ref{lemma:total_rate_aver_dist}. In particular, if the temporally correlated joint $\pmf$ in \eqref{source_correlated_1} follows a finite-order Markov process, then, the description of the rate-distortion region in Lemma \ref{lemma:coding_theorem_iid}, and the corresponding bounds on the minimum achievable total-rate in Lemmas \ref{lemma:total_rate}, \ref{lemma:total_rate_aver_dist} can be simplified considerably following for instance the framework of \cite{witsenhausen:1979,borkar:2001,stavrou:2018siam}. For the important special case of first-order Markov process, \eqref{single_letter_char_iid} simplifies to
\begin{align}
{\cal R}^{\IID,1}&=\Bigg\{(R_{1,n},D_{1,n})\Bigg{|}\exists S_{1,n-1}, Y_{1,n},~\{g_t(\cdot)\}_{t=1}^n,\nonumber\\
s.t.&~{R}_1\geq{I}(X_1;S_1),~\mbox{(initial time)}\nonumber\\
&~R_t\geq{I}(X_{t};S_t|S_{1,t-1}),~t=2,\ldots,n-1,\nonumber\\
& ~R_n\geq{I}(X_{n};Y_n|S_{1,n-1}),~\mbox{(terminal time)},\nonumber\\
& ~D_t\geq{\bf E}\left\{d_t(X_t,Y_t)\right\},~t\in\mathbb{N}_1^n,\nonumber\\
& ~Y_1=g_1(S_1),~Y_{t}=g_t(S_{1,t}),~t=2,\ldots,n-1,\nonumber\\
& ~S_1\leftrightarrow(X_{1})\leftrightarrow{X}_{2,n},\nonumber\\
& ~S_t\leftrightarrow(X_{t},S_{1,t-1})\leftrightarrow(X_{1,t-1},{X}_{t+1,n})\Bigg\}.
\label{single_letter_char_iid_markov}
\end{align}
Using \eqref{single_letter_char_iid_markov}, the minimum achievable total-rate can now be simplified to the following optimization problem:
\begin{align}
{\cal R}_{\ssum}^{\IID,\op,1}(D_{1,n})\triangleq\min_{(R_{1,n},D_{1,n})\in{\cal R}^{\IID,1}}\sum_{t=1}^nR_t.\label{operational_min_sum_rate_new_markov}
\end{align}
Using the description  of  \eqref{operational_min_sum_rate_new_markov}, we can simplify \eqref{eq:sumrate} and \eqref{eq:sumrate_aver_dist}, respectively,  as follows:
\begin{align}
&{\cal R}_{\ssum}^{\IID,\op,1}(D_{1,n})\geq{\cal R}_{{\ssum}}^{\IID,1}(D_{1,n})=\min_{\substack{{\bf E}\left\{d_t(X_t,Y_t)\right\}\leq{D}_t,~t\in\mathbb{N}_1^n\\Y_1\leftrightarrow{X_{1}}\leftrightarrow{X}_{2,n},\\~Y_t\leftrightarrow(X_{t},Y_{1,t-1})\leftrightarrow(X_{1,t-1},{X}_{t+1,n}),~t\in\mathbb{N}_2^{n-1}}}I(X_{1,n};Y_{1,n}),\label{eq:sumrate1}\\
&{\cal R}_{{\ssum}}^{\IID,\op,1}(D)\geq{\cal R}_{{\ssum}}^{\IID,1}(D)\triangleq\min_{\substack{\frac{1}{n}\sum_{t=1}^n{\bf E}\left\{d_t(X_t,Y_t)\right\}\leq{D},~t\in\mathbb{N}_1^n\\Y_1\leftrightarrow{X_{1}}\leftrightarrow{X}_{2,n},\\~Y_t\leftrightarrow(X_{t},Y_{1,t-1})\leftrightarrow(X_{1,t-1},~{X}_{t+1,n}),~t=2,\ldots,n-1}}I(X_{1,n};Y_{1,n}),\label{eq:sumrate2}
\end{align}
where $I(X_{1,n};Y_{1,n})=\sum_{t=1}^n{I}(X_t;Y_t|Y_{1,t-1})$.}
\par {In the sequel, we use the description of \eqref{eq:sumrate2} to derive our main results.}
%
%
%
%
\section{Application in Quantized State Estimation}\label{sec:examples}

\par In this section, {we apply the sequential coding framework of the previous section to a state estimation problem and obtain new results in such applications.}
\par We consider a similar scenario to \cite[\S{II}]{khina:2018} where a multi-track system  estimates several ``parallel'' Gaussian processes over a single shared communication link as illustrated in Fig. \ref{fig:multitrack}. Following the sequential coding framework, we require the Gaussian source processes to have temporally correlated and spatially $\IID$ components, which are observed by an observer who collects the measured states into a single vector state. Then, the observer/encoder maps the states as random finite-rate packets to a $\mmse$ estimator through a noiseless link. {Compared to the result of \cite[Theorem 1]{khina:2018} which derives a dynamic forward in time recursion of a distortion-rate allocation algorithm when the rate is given at each time instant, here we derive a dynamic rate-distortion reverse-waterfilling algorithm operating forward in time for which we only consider a given distortion threshold $D>0$.}  
\begin{figure}[htp]
\centering
\includegraphics[width=\textwidth]{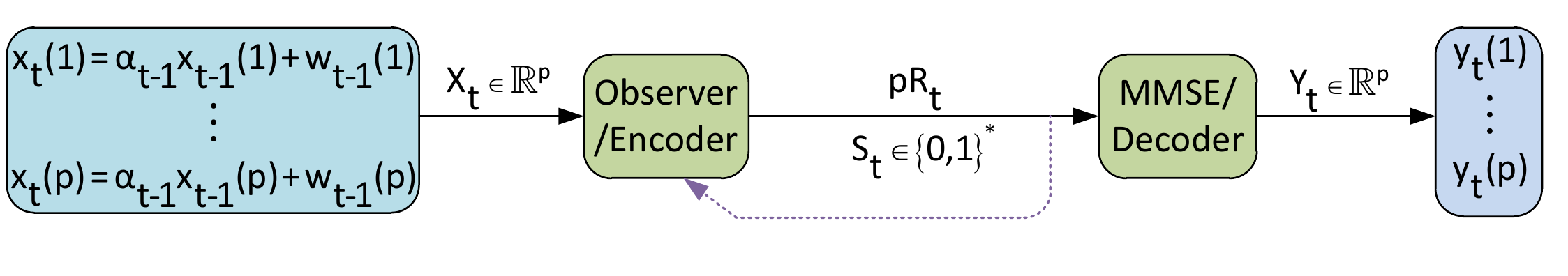}
\caption{Multi-track state estimation system model.} 
\label{fig:multitrack}
\end{figure}
\par First, we describe the problem of interest.\\
{\bf State process.} Consider $p$-parallel time-varying Gauss-Markov processes with $\IID$ spatial components as follows:
\begin{align}
x_{t}(i)=\alpha_{t-1}x_{t-1}(i)+w_{t-1}(i),~i\in\mathbb{N}_1^p,~t\in\mathbb{N}_1^n,\label{example:GM_process}
\end{align}
where $x_1(i)\equiv{x}_1$ is given, with $x_1\sim{\cal N}(0;\sigma^2_{x_1})$; the non-random coefficient $\alpha_{t}\in\mathbb{R}$ is known at each time step $t$, and $\{w_t(i)\equiv{w}_t:~i\in\mathbb{N}_1^p\}$, $w_t\sim{\cal N}(0;\sigma^2_{w_{t}})$, is an independent Gaussian noise process at each $t$, independent of $x_1, \forall{i}\in\mathbb{N}_1^p$. Since \eqref{example:GM_process} has $\IID$ spatial components it can be compactly written as a vector or frame as follows:
\begin{align}
X_{t}=A_{t-1}X_{t-1}+W_{t-1},~X_1=\text{given},~t\in\mathbb{N}_2^n,\label{example:GM_process1}
\end{align}
where $A_{t-1}=\diag(\alpha_{t-1},\ldots,\alpha_{t-1})\in\mathbb{R}^{p\times{p}}$, $X_t\in\mathbb{R}^p$, and the independent Gaussian noise process $W_t\in\mathbb{R}^p\sim{\cal N}(0;\Sigma_{W_t})$, where $\Sigma_{W_t}=\diag(\sigma^2_{w_t},\ldots,\sigma^2_{w_t})\succ{0}\in\mathbb{R}^{p\times{p}}$ independent of the initial state $X_1$.\\ 
{\bf Observer/Encoder.} At the observer the spatially $\IID$ time-varying $\mathbb{R}^p$-valued Gauss-Markov processes are collected into  a frame $X_t\in\mathbb{R}^p$ and mapped using sequential coding with encoded sequence:
\begin{align}
S_t=f_t(X_{1,t},S_{1,t-1}),\label{example:encoding}
\end{align}
where at $t=1$ we assume $S_1=f_1(X_1)$, and $R_t=\frac{{\bf E}|S_t|}{p}$ is the expected (random) rate (per dimension) at each time instant $t$ transmitted through the noiseless link.\\
{\bf \text{MMSE} Decoder.} The data packet $S_t$ is received using the following reconstructed sequence:
\begin{align}
Y_t=g_t(S_{1,t}),\label{example:decoding}
\end{align}
where at $t=1$ we have $Y_1=g_1(S_1)$.\\
{\bf Distortion.} We consider the average total $\mse$ distortion normalized over all spatial components as follows:
\begin{align}
\frac{1}{n}\sum_{t=1}^nD_t~\mbox{with}~D_t\triangleq\frac{1}{p}{\bf E}\left\{||X_t-Y_t||_2^2\right\}.\label{example:aver_dist}
\end{align} 
\noindent{\bf Performance.} The performance of the above system {(per dimension)} for a given $D>0$ can be cast to the following optimization problem:
\begin{align}
{\cal R}_{\ssum}^{\IID,\op,1}(D)=\min_{\substack{(f_t,~g_t):~t=1,\ldots,n\\ \frac{1}{n}\sum_{t=1}^nD_t\leq{D}}}\sum_{t=1}^n{R}_t.\label{example:problem_form}
\end{align}
\par  {The next theorem is our first main result in this paper. It derives a lower bound on the performance of Fig. \ref{fig:multitrack} by means of a dynamic reverse-waterfilling algorithm.} 
\begin{theorem}(Lower bound on \eqref{example:problem_form})\label{theorem:fundam_lim_mmse}
For the multi-track system in Fig. \ref{fig:multitrack}, the minimum achievable total-rate for any ``$n$'' and any $p$, however large, is ${\cal R}_{\ssum}^{\IID,\op,1}(D)=\sum_{t=1}^n{R^{\op}_t}$ with the minimum achievable rate distortion at each time instant (per dimension) given by some $R^{\op}_t\geq{R}^*_t$ such that
\begin{align}
R_t^*=\frac{1}{2}\log_2\left(\frac{\lambda_t}{D_{t}}\right), \label{para_sol_eq.1}
\end{align}
where $\lambda_t\triangleq\alpha^2_{t-1}D_{t-1}+\sigma^2_{w_{t-1}}$ and $D_t$ is the distortion at each time instant evaluated based on a dynamic reverse-waterfilling algorithm operating forward in time. The algorithm is as follows:
\begin{align}
D_{t} &\triangleq\left\{ \begin{array}{ll} \xi_t  & \mbox{if} \quad \xi_t\leq\lambda_{t} \\
\lambda_{t} &  \mbox{if}\quad\xi_t>\lambda_{t} \end{array} \right.,~\forall{t},\label{rev_water:eq.1}
\end{align}
with $\sum_{t=1}^n D_t= nD$, and 
\begin{align}\label{values_of_xis}
\xi_t=
\begin{cases}
\frac{1}{2b^2_t}\left(\sqrt{1+\frac{2b^2_t}{\theta}}-1\right),~\forall t\in\mathbb{N}_1^{n-1} \\
\frac{1}{2\theta},~t=n
\end{cases},
\end{align}
where $\theta>0$ is the Lagrangian multiplier tuned to obtain equality $\sum_{t=1}^n D_t= nD$, $b^2_t\triangleq\frac{\alpha^2_t}{\sigma^2_{w_t}}$, and  $D\in(0, \infty)$.
\end{theorem}
\begin{IEEEproof} 
See Appendix \ref{proof:theorem:fundam_lim_mmse}.
\end{IEEEproof}

{In the next remark, we discuss some technical observations regarding Theorem \ref{theorem:fundam_lim_mmse} and draw connections with \cite[Corollary 1.2]{ma-ishwar:2011}.}
{\begin{remark}\label{remark:discussion:theorem:mmse}
{\bf (1)} The optimization problem in the derivation of Theorem \ref{theorem:fundam_lim_mmse} suggests that $(A_t, \Sigma_{W_t}, \Delta_t, \Lambda_t)$ commute by pairs\cite[p. 5]{harville:1997} since they are all scalar matrices which in turn means that they are simultaneously diagonalizable by an orthogonal matrix \cite[Theorem 21.13.1]{harville:1997} (in this case the orthogonal matrix is the identity matrix hence it is omitted from the characterization of the optimization problem). \\
{\bf (2)} Theorem \ref{theorem:fundam_lim_mmse} extends the result of \cite[Corollary 1.2]{ma-ishwar:2011} who found an explicit expression of the minimum total-rate $\sum_{t=1}^nR_t^*$ for $n=3$ subject to a per-time $\mse$ distortion, to a similar problem constrained by an average total-distortion that we solve using a dynamic reverse-waterfilling algorithm that allocates the rate and the distortion at each instant of time for a fixed finite time horizon. 
\end{remark}
}

{\it Implementation of the dynamic reverse-waterfilling:} It should be remarked that a way to implement the reverse-waterfilling algorithm in Theorem \ref{theorem:fundam_lim_mmse} is proposed in \cite[Algorithm 1]{stavrou:2018lcss}. A different algorithm using the {\it bisection method }(for details see, e.g., \cite[Chapter 2.1]{atkinson:1991}) is proposed in Algorithm \ref{algo1}. {The method in Algorithm \ref{algo1} guarantees linear convergence with rate $\frac{1}{2}$. On the other hand, \cite[Algorithm 1]{stavrou:2018lcss} requires a specific proportionality gain factor $\gamma\in(0,1]$ chosen appropriately at each time instant. The choice of $\gamma$ affects the rate of convergence whereas it does not guarantee global convergence of the algorithm}. In Fig. \ref{fig:rate-distortion:allocation}, we illustrate a numerical simulation using Algorithm \ref{algo1} by taking $a_t\in(0,2)$, $\sigma^2_{w_t}=1,$~for $t=\{1,2,\ldots,200\}$ and $D=1$. 

\begin{figure}[htp]
\centering
\includegraphics[width=\columnwidth]{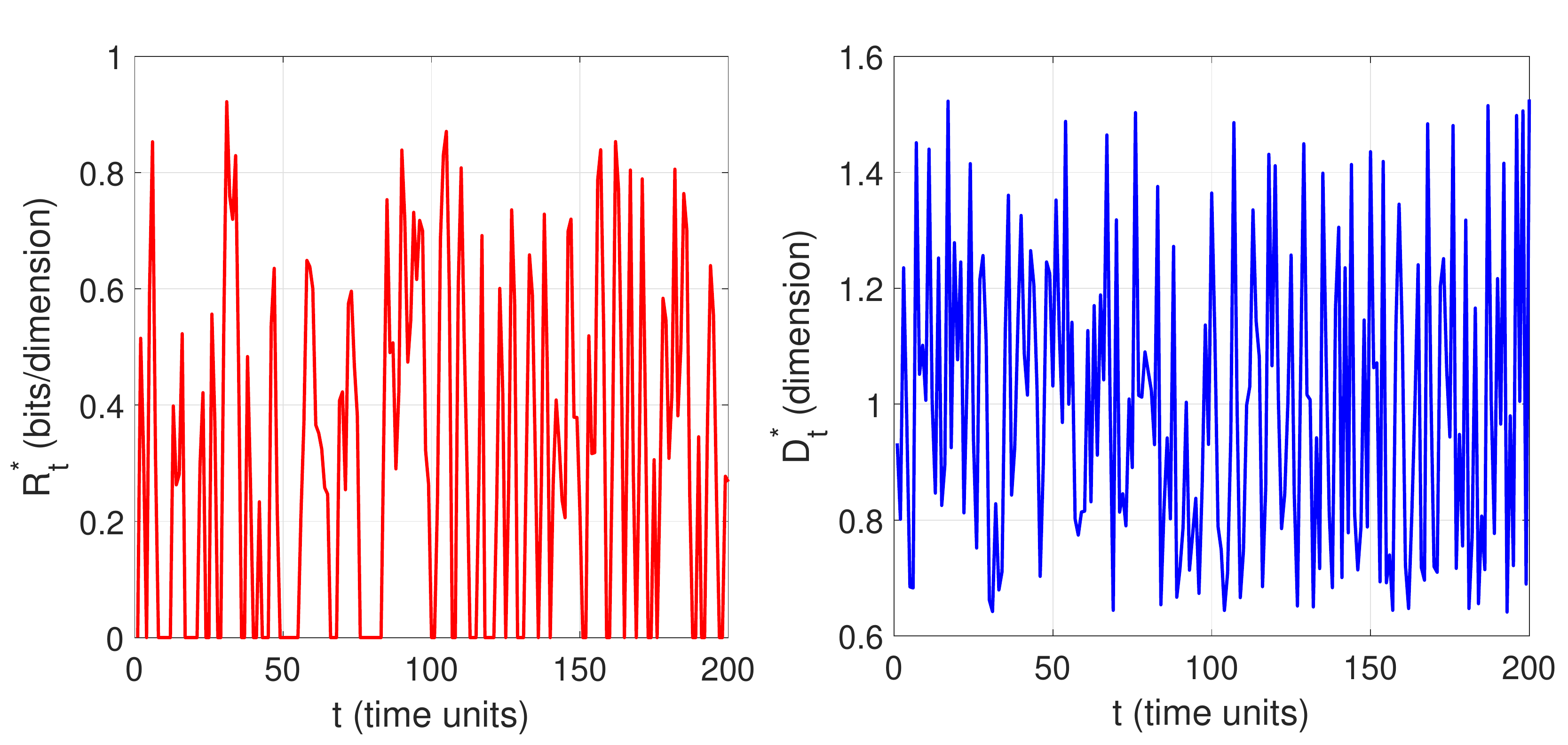}
\caption{Dynamic rate-distortion allocation for a time-horizon $t=\{1,2\ldots,200\}$ for the system in Fig. \ref{fig:multitrack}.} 
\label{fig:rate-distortion:allocation}
\end{figure}
\begin{varalgorithm}{1}
\caption{Dynamic reverse-waterfilling algorithm}
\begin{algorithmic}
\STATE {\textbf{Initialize:} number of time-steps $n$; distortion level $D$; error tolerance $\epsilon$; nominal minimum and maximum value $\theta^{\min}={0}$ and $\theta^{\max}={\frac{1}{2D}}$; initial variance ${\lambda}_1=\sigma^2_{x_{1}}$ of the initial state $x_1$, values $a_t$ and $\sigma^2_{w_t}$ of  \eqref{example:GM_process}.} 
\STATE {Set $\theta=1/2D$; $\text{flag}=0$.}
\WHILE{$\text{flag}=0$}
\STATE {Compute $D_{t}~\forall~t$ as follows:}
\FOR {$t=1:n$}
\STATE {Compute $\xi_t$ according to \eqref{values_of_xis}.}
\STATE {Compute $D_t$ according to \eqref{rev_water:eq.1}.}
\IF {$t<n$}
\STATE{Compute $\lambda_{t+1}$ according to $\lambda_{t+1}\triangleq\alpha^2_{t}D_{t}+\sigma^2_{w_{t}}$.}
\ENDIF
\ENDFOR
\IF {$\frac{1}{n}\sum{D}_t-D\geq\epsilon$}
\STATE {Set $\theta^{\min}=\frac{\theta}{n}$.}
\ELSE
\STATE {Set $\theta^{\max}=\frac{\theta}{n}$.}
\ENDIF
\IF {$\theta^{\max}-\theta^{\min}\geq\frac{\epsilon}{n}$}
\STATE{Compute $\theta=\frac{n(\theta^{\min}+\theta^{\max})}{2}$.}
\ELSE
\STATE {$\text{flag}\leftarrow 1$}
\ENDIF
\ENDWHILE
\STATE {\textbf{Output:} $\{D_t:~t\in\mathbb{N}_1^n\}$, $\{\lambda_t:~t\in\mathbb{N}_1^n\}$, for a given distortion level $D$.}
\end{algorithmic}
\label{algo1}
\end{varalgorithm}

{
\subsection{Steady-state solution of Theorem \ref{theorem:fundam_lim_mmse}}\label{subsec:steady-state_solution}}

{In this subsection, we study the steady-state case of the lower bound obtained in Theorem \ref{theorem:fundam_lim_mmse}. To do this, first, we restrict the state process of our setup to be time invariant, which means that in \eqref{example:GM_process} the coefficients $\alpha_{t-1}\equiv{\alpha},~\forall{t}$ and $w_t\sim{\cal N}(0;\sigma^2_{w}),~\forall{t}$, or similarly, in \eqref{example:GM_process1} the matrix $A_{t-1}\equiv{A}=\diag(\alpha,\ldots,\alpha),~\forall{t}$ and $W_t\sim{\cal N}(0;\Sigma_W),~\forall{t}$, where $\Sigma_W=\diag(\sigma^2_{w},\ldots,\sigma^2_{w})\succ{0}$.
We also denote the steady-state average total rate and distortion as follows:
\begin{align}
R_{\infty}=\limsup_{n\longrightarrow\infty}\frac{1}{n}\sum_{t=1}^nR_t,~~~~D_{\infty}=\limsup_{n\longrightarrow\infty}\frac{1}{n}\sum_{t=1}^nD_t.\label{steady_state_rate_distortion}
\end{align}
\noindent{\bf Steady-state Performance.} The minimum achievable steady-state performance of the multi-track system of Fig. \ref{fig:multitrack} when the system is modeled by $p$-parallel time-invariant Gauss-Markov processes {(per dimension)} can be cast to the following optimization problem:
\begin{align}
{\cal R}_{\ssum,ss}^{\IID,\op,1}(D)=\min_{\substack{(f_t,~g_t):~t=1,\ldots,\infty\\ D_{\infty}\leq{D}}}R_{\infty}.\label{performance_steady_state}
\end{align}
The next corollary is a consequence of the lower bound derived in Theorem \ref{theorem:fundam_lim_mmse}. It states that the minimum achievable steady state total rate subject to steady-state total distortion constraint is equivalent to having the minimum achievable steady state total rate subject to a fixed distortion budget, i.e., $D_t=D,~\forall{t}$.  This result complements  equivalent results derived in \cite{derpich:2012}, \cite[Corollary 2]{khina:2018}.}
{
\begin{corollary}(Lower bound on \eqref{performance_steady_state})\label{corollary:steady_state_rev_water}
The minimum achievable steady state performance of \eqref{performance_steady_state}, under a steady-state total distortion constraint $D_{\infty}\leq{D}$ for any $p$ however larger, is bounded from below by ${\cal R}_{\ssum,ss}^{\IID,\op,1}(D)\geq{R}^*_{\infty}$, such that
\begin{align}
R^*_{\infty}=\frac{1}{2}\log_2\left(\alpha^2+\frac{\sigma^2_w}{D}\right),\label{steady_state_rdf}
\end{align}
where $R^*_{\infty}\triangleq\lim_{n\longrightarrow\infty}\frac{1}{n}\sum_{t=1}^nR_t^*$. Consequently, assuming $D_t=D,~\forall{t}$, achieves \eqref{steady_state_rdf} as $n\longrightarrow\infty$.
\end{corollary}
}
{
\begin{IEEEproof}
To obtain our result, we first take the average total-rate, i.e., $\frac{1}{n}\sum_{t=1}^nR_t$. Then, we show the following inequalities:
\begin{align}
\frac{1}{n}\sum_{t=1}^nR_t&\stackrel{(a)}\geq\frac{1}{n}\sum_{t=1}^nR_t^*\nonumber\\
&=\frac{1}{n}\sum_{t=1}^n\frac{1}{2}\log_2\left(\frac{\lambda_t}{D_t}\right)\nonumber\\
&\stackrel{(b)}=\frac{1}{n}\sum_{t=1}^n\left[\frac{1}{2}\log_2\left(\alpha^2{D}_{t-1}+\sigma^2_w\right)-\frac{1}{2}\log_2{D_t}\right]\nonumber\\
&=\frac{1}{2n}\log_2\left(\frac{\lambda_1}{D_n}\right)+\frac{1}{2n}\sum_{t=1}^{n-1}\log_2\left(\alpha^2+\frac{\sigma^2_w}{D_t}\right)\nonumber\\
&\stackrel{(c)}\geq\frac{1}{2n}\log_2\frac{\lambda_1}{\sum_{t=1}^nD_t}+\frac{1}{2}\frac{(n-1)}{n}\log_2\left(\alpha^2+\frac{(n-1)\sigma^2_w}{\sum_{t=1}^{n-1}D_t}\right)\nonumber\\
&\stackrel{(d)}\geq\frac{1}{2n}\log_2\left(\frac{\lambda_1}{nD}\right)+\frac{1}{2}\frac{(n-1)}{n}\log_2\left(\alpha^2+\frac{\sigma^2_w}{\frac{nD}{n-1}}\right),
\label{derivation_steady_state}
\end{align}
where $(a)$ follows from Theorem \ref{theorem:fundam_lim_mmse}; $(b)$ follows because for time-invariant processes $\lambda_t=\alpha^2D_{t-1}+\sigma^2_{w}$; $(c)$ follows because in the first term  $D_n\leq\sum_{t=1}^nD_t$ and in the second term we apply Jensen's inequality \cite[Theorem 2.6.2]{cover-thomas2006}; $(d)$ follows because in the first term $\sum_{t=1}^{n}D_t\leq{nD}$ and in the second term  $\sum_{t=1}^{n-1}D_t\leq\sum_{t=1}^{n}D_t\leq{nD}$ since $D_n\geq{0}$. {We prove that ${\cal R}_{\ssum,ss}^{\IID,\op,1}(D)\geq{R}^*_{\infty}$ where $R^*_{\infty}$ is given by \eqref{steady_state_rate_distortion} by evaluating \eqref{derivation_steady_state} in the limit $n\longrightarrow\infty$ and then minimizing both sides. This is obtained because the first term equals to zero ($\lambda_1, D$ are constants), in the second term $\lim_{n\longrightarrow\infty}\left(\frac{n-1}{n}\right)=1$, $\lim_{n\longrightarrow\infty}\left(\frac{n}{n-1}\right)=1$ and then by taking the minimization in both sides. This completes the proof.}
\end{IEEEproof}
\begin{remark}(Connections to existing works)\label{remark:existing_nrdf}
We note that the steady-state lower bound (per dimension) obtained in Corollary \ref{corollary:steady_state_rev_water} corresponds precisely to the solution of the time-invariant scalar-valued Gauss-Markov processes with per-time MSE distortion constraint derived in \cite[Equation (14)]{tatikonda:2004} and to the solution of stationary Gauss-Markov processes with MSE distortion constraint derived in \cite[Theorem 3]{derpich:2012}, \cite[Equation (1.43)]{gorbunov-pinsker1972b}. 
\end{remark}}
\subsection{Upper bounds to the minimum achievable total-rate}\label{subsec:upper_bound_mmse} 

In this section, we employ a sequential causal $\dpcm$-based scheme using pre/post filtered $\ecdq$ (for details, see, e.g., \cite[Chapter 5]{zamir:2014}) that ensures standard performance guarantees (achievable upper bounds) on the minimum achievable sum-rate ${\cal R}_{\ssum}^{\IID,\op,1}(D)=\sum_{t=1}^nR^{\op}_t$ of the multi-track setup of Fig. \ref{fig:multitrack}. {The reason for the choice of this quantization scheme is twofold. First, it can be implemented in practice and, second, it allows to find analytical achievable bounds and approximations on finite-dimensional quantizers which generate near-Gaussian quantization noise and Gaussian quantization noise for infinite dimensional quantizers \cite{zamir-feder:1996b}.}

\par  {We first describe the sequential causal $\dpcm$ scheme using $\mmse$ quantization for parallel time-varying Gauss-Markov processes. Then, we bound the rate performance of such scheme using $\ecdq$ and vector quantization followed by memoryless entropy coding. This can be seen as a generalization of \cite[Corollary 1.2]{ma-ishwar:2011} to any finite time when the rate is allocated at each time instant. Observe that because the state is modeled as a first-order Gauss-Markov process, the sequential causal coding is precisely equivalent to predictive coding (see, e.g., \cite{stavrou:2018stsp},~\cite[Theorem 3]{yang:2014}). Therefore, we can immediately apply the standard sequential causal $\dpcm$ \cite{farvardin:1985,jayant:1990} approach (with $\mathbb{R}^p$-valued $\mmse$ quantizers) to obtain an achievable rate in our system. \\
{\bf DPCM scheme.} At each time instant $t$ the {\it encoder or innovations' encoder} performs the linear operation 
\begin{align}
\widehat{X}_t=X_t-A_{t-1}Y_{t-1},\label{dpcm_enc}
\end{align}
where at $t=1$ we have $\widehat{X}_{1}=X_1$ and also $Y_{t-1}\triangleq{\bf E}\left\{X_{t-1}|S_{1,t-1}\right\}$, i.e., an estimate of $X_{t-1}$ given the previous quantized symbols $S_{1,t-1}$.\footnote{Note that the process $\widehat{X}_t$ has a temporal correlation since it is the error of $X_t$ from all quantized symbols $S_{1,t-1}$ and not the infinite past of the source $X_{-\infty,t}=(X_{-\infty},\ldots,X_t)$. Hence, $\widehat{X}_t$ is only an estimate of the true process.} Then, by means of a $\mathbb{R}^p$-valued $\mmse$ quantizer that operates at a rate (per dimension) $R_t$, we generate the quantized reconstruction $\widehat{Y}_t$ of the residual source $\widehat{X}_t$ denoted by $\widehat{Y}_t=Y_t-A_{t-1}Y_{t-1}$. Then, we send $S_t$ over the channel (the corresponding data  packet to $\widehat{Y}_t$). At the {\it decoder} we receive $S_t$ and recover the quantized symbol $\widehat{Y}_t$ of $\widehat{X}_t$.}
\begin{figure}[htp]
\centering
\includegraphics[width=\columnwidth]{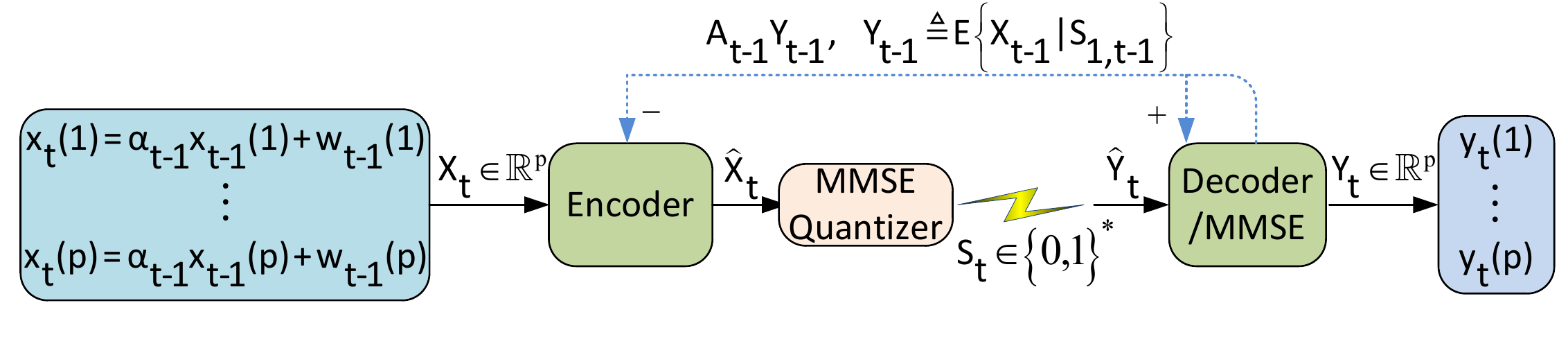}
\caption{{$\dpcm$ of parallel processes.}} 
\label{fig:dpcm}
\end{figure} 
Then, we generate the estimate $Y_t$ using the linear operation 
\begin{align}
Y_t=\widehat{Y}_t+A_{t-1}Y_{t-1}.\label{dpcm_dec}
\end{align}
Combining both \eqref{dpcm_enc}, \eqref{dpcm_dec}, we obtain
\begin{align}
X_t-Y_t=\widehat{X}_t-\widehat{Y}_t.\label{data_processing}
\end{align}
{\it MSE Performance.} From \eqref{data_processing}, we see that the error between $X_t$ and $Y_t$ is equal to the quantization error introduced by $\widehat{X}_t$ and $\widehat{Y}_t$. This also means that the $\mse$ distortion (per dimension) at each instant of time satisfy
\begin{align}
D_t=\frac{1}{p}{\bf E}\{||X_t-Y_t||_2^2\}=\frac{1}{p}{\bf E}\{||\widehat{X}_t-\widehat{Y}_t||_2^2\}.\label{dpcm_mse}
\end{align}
 A pictorial view of the $\dpcm$ scheme is given in Fig. \ref{fig:dpcm}.

{The following theorem is another main result of this section.}
\begin{theorem}(Upper bound to $R_{\ssum}^{\IID,\op,1}(D)$)\label{theorem:achievability}
Suppose that in \eqref{example:problem_form} we apply a sequential causal $\dpcm$-based $\ecdq$ with a lattice quantizer. Then, the minimum achievable total-rate ${\cal R}_{\ssum}^{\IID,\op,1}(D)=\sum_{t=1}^n{R}^{\op}_t$, where at each time instant ${R}^{\op}_t$ is upper bounded as follows:
\begin{align}
R^{\op}_t\leq{R}^*_t+\frac{1}{2}\log_2\left(2\pi{e}G_p\right)+\frac{1}{p},~\forall{t}, ~\mbox{(bits/dimension)},\label{achievable_bound_ecdq}
\end{align} 
where $R_t^*$ is obtained from Theorem \ref{theorem:fundam_lim_mmse}, $\frac{1}{2}\log_2\left(2\pi{e}G_p\right)$ is the divergence of the quantization noise from Gaussianity; $G_p$ is the dimensionless normalized second moment of the lattice \cite[Definition 3.2.2]{zamir:2014} and $\frac{1}{p}$ is the additional cost due to having prefix-free (instantaneous) coding.
\end{theorem}
\begin{IEEEproof}
See Appendix~\ref{proof:theorem:achievability}.
\end{IEEEproof}

\par {Next, we remark some technical comments related to Theorem \ref{theorem:achievability}, to better explain its novelty compared to the existing similar schemes in the literature.}

\begin{remark}(Comments on Theorem \ref{theorem:achievability})\label{remark:comments_on_theorem_achievability}
{\bf (1)} {The bound of Theorem \ref{theorem:achievability} allows the transmit rate to vary at each time instant for a finite time horizon while it achieves the $\mmse$ distortion at each time step $t$. This is because our $\dpcm$-based $\ecdq$ scheme is constrained by total-rates that we find at each instant of time using the dynamic reverse-waterfilling algorithm of Theorem  \ref{theorem:fundam_lim_mmse}. This loose rate-constraint is the new input of our bound compared to similar existing bounds in the literature (see, e.g., \cite[Theorem 6, Remark 16]{khina:2018}, \cite[Corollary 5.2]{silva:2011}, \cite[Theorem 5]{stavrou:2018stsp}) that assume {\it fixed rates averaged across the time or asymptotically average total rate constraints} hence restricting their transmit  rate at each instant of time to be the same for any time horizon.}\\
{\bf (2)} Recently in \cite{kostina:2019} (see also \cite{khina:2018}), it is pointed out that for discrete-time processes one can assume in the $\ecdq$ coding scheme the clocks of the entropy encoder and the entropy decoder to be synchronized, thus, eliminating the additional rate-loss due to prefix-free coding. This assumption, will give a better upper bound in Theorem \ref{theorem:achievability} because the term $\frac{1}{p}$ will be removed.
\end{remark}

\noindent{{\bf Steady-state performance.} Next, we describe how to obtain an upper bound on \eqref{performance_steady_state}. Suppose that the system is modeled by $p$-parallel time-invariant Gauss-Markov processes {(per dimension)} similar to \S\ref{subsec:steady-state_solution}.}
{
\begin{corollary}(Upper bound on \eqref{performance_steady_state})\label{corollary:achievability_ss}
Suppose that in \eqref{example:problem_form} we apply a sequential causal $\dpcm$-based $\ecdq$ with a lattice quantizer assuming the system is time-invariant and that $D_t=D$, $\forall{t}$. Then, the minimum achievable steady-state performance ${\cal R}_{\ssum,ss}^{\IID,\op,1}(D)={R}^{\op}_{\infty}$ is upper bounded as follows:
\begin{align}
{R}^{\op}_{\infty}\leq{R}^*_\infty+\frac{1}{2}\log_2\left(2\pi{e}G_p\right)+\frac{1}{p},~\mbox{(bits/dimension)},\label{achievable_bound_ecdq_steadu_state}
\end{align} 
where $R_\infty^*$ is given by \eqref{steady_state_rdf}.
\end{corollary}
\begin{IEEEproof}
This follows from  Theorem \ref{theorem:achievability} and Corollary \ref{corollary:steady_state_rev_water}. 
\end{IEEEproof}
We note that Corollary \ref{corollary:achievability_ss} is a known infinite time horizon bound derived in several paper in the literature, such as those discussed in Remark \ref{remark:comments_on_theorem_achievability}, {\bf (1)}. }

\paragraph*{\it Computation of Theorem \ref{theorem:achievability}} Unfortunately, finding $G_p$ in \eqref{achievable_bound_ecdq} for good high-dimensional quantizers {of possibly finite dimension} is currently an open problem (although it can be approximated for any dimension using for example product lattices \cite{conway-sloane1999}). Therefore, in what follows we propose existing computable bounds to the achievable upper bound of Theorem \ref{theorem:achievability} for any high-dimensional lattice quantizer. 
Note that these bounds were derived as a consequence of the main result by Zador \cite{conway-sloane1999}, namely, it is possible to reduce the $\mse$ distortion normalized per dimension using higher-dimensional quantizers. Toward this end, Zador introduced a lower bound on $G_p$ using the dimensionless normalized second moment of a $p$-dimensional sphere, hereinafter denoted by $G(S_p)$, for which it holds that:
\begin{align}
G(S_p)=\frac{1}{(p+2)\pi}\Gamma\left(\frac{p}{2}+1\right)^{\frac{2}{p}},\label{zador_lower_bound}
\end{align} 
where $\Gamma(\cdot)$ is the gamma function. Moreover, $G_p$ and $G(S_p)$ are connected via the following  inequalities:
\begin{align}
\frac{1}{2\pi{e}}\sr{(a)}\leq{G}(S_p)\sr{(b)}\leq{G}_p\sr{(c)}\leq\frac{1}{12},\label{connection_sphere_packing}
\end{align}
where $(a), (b)$ holds with equality for $p\longrightarrow\infty$; $(c)$ holds with equality if $p=1$.\\
Note that in \cite[equation (82)]{conway-sloane1999}, there is also an upper bound on $G_p$ due to Zador. The bound is the following:
\begin{align}
G_p\leq\frac{1}{p\pi}\Gamma\left(\frac{p}{2}+1\right)^{\frac{2}{p}}\Gamma\left(1+\frac{2}{p}\right).\label{zador_upper_bound}
\end{align}

%

\begin{figure}[ht]
\centering
\begin{subfigure}[b]{0.8\columnwidth}
{\includegraphics[width=\columnwidth]{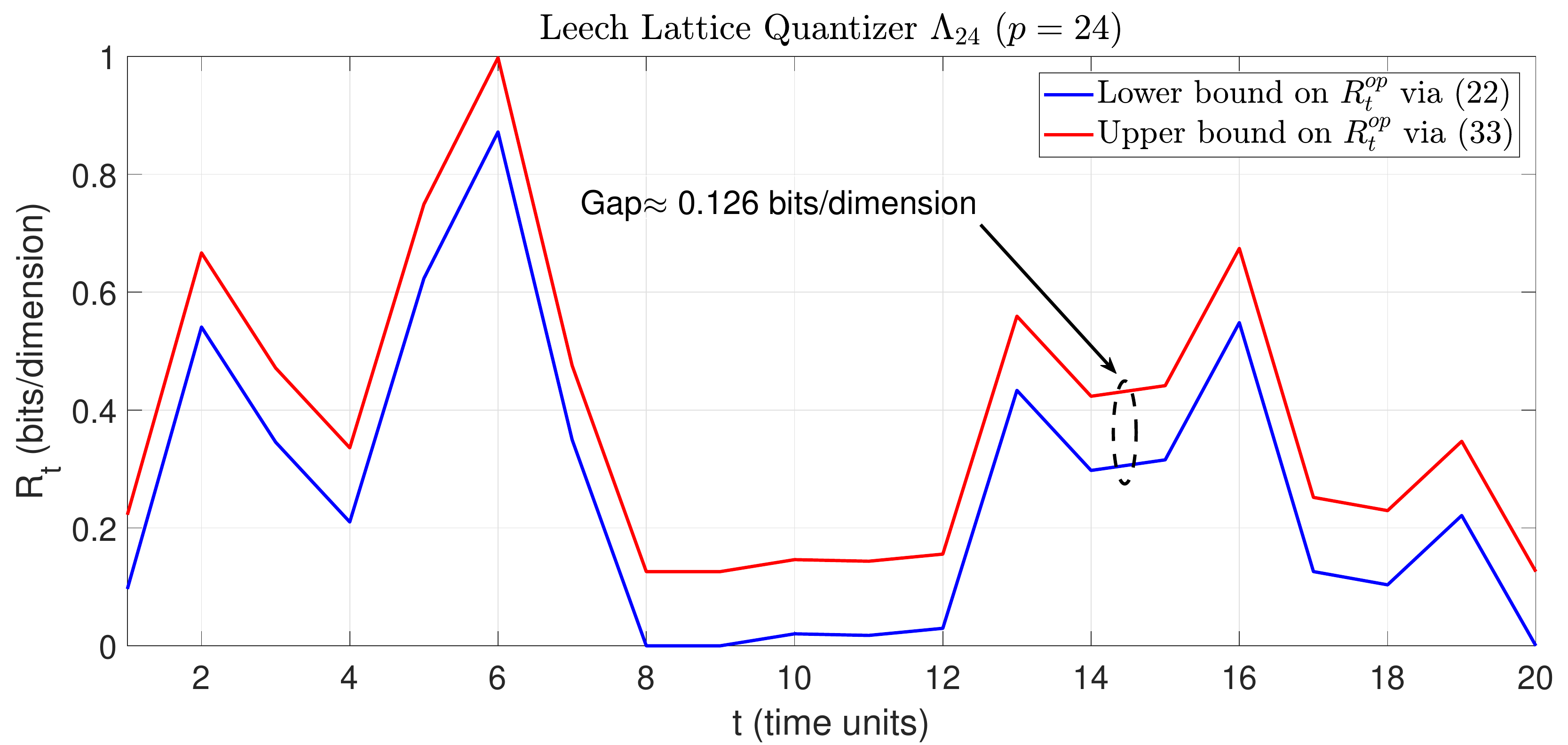}}
\end{subfigure}
\begin{subfigure}[b]{0.8\columnwidth}
{\includegraphics[width=\columnwidth]{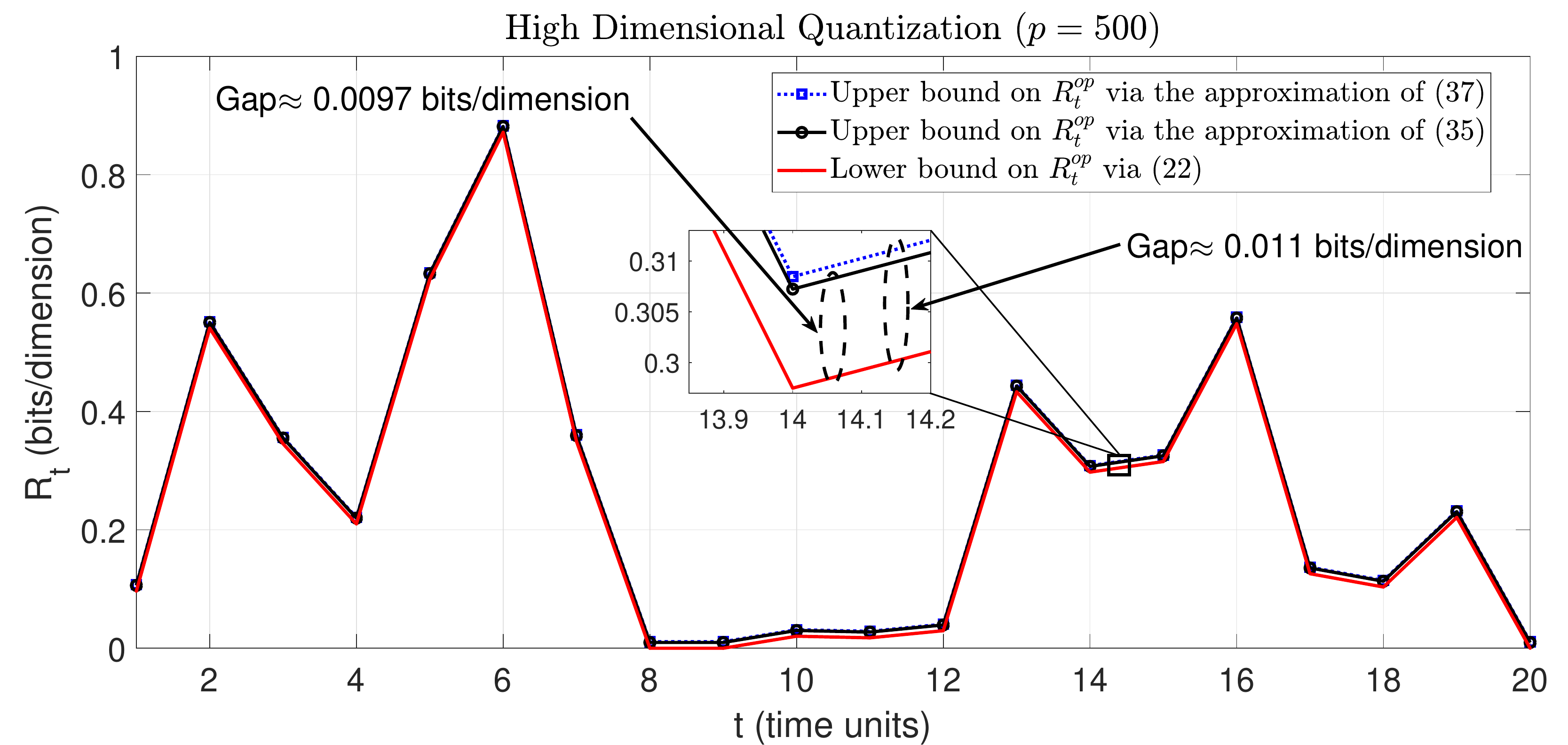}}
\end{subfigure}
\caption{Bounds on the minimum achievable total-rate.}
\label{fig:oper_rates}
\end{figure}
In Fig. \ref{fig:oper_rates} we illustrate two plots where we compute the bounds derived in Theorems \ref{theorem:fundam_lim_mmse}, \ref{theorem:achievability} for two different scenarios. 
{In Fig.~\ref{fig:oper_rates}, {(a)}, we choose $t=\{1,\ldots,20\}$, $a_t\in(0,1.5)$,~$\sigma^2_{w_t}=1$, and~$D=1$, to illustrate the gap between the time-varying rate-distortion allocation obtained using the lower bound \eqref{para_sol_eq.1} and the upper bound \eqref{achievable_bound_ecdq} when the latter is approximated with the best known quantizer up to twenty four dimensions that is a lattice known as Leech lattice quantizer (for details see, e.g., \cite[Table 2.3]{conway-sloane1999}). For this experiment the gap between the two bounds is approximately $0.126$ bits/dimension. In Fig.~\ref{fig:oper_rates}, {(b)}, we perform another experiment assuming the same values for ($a_t,~\sigma^2_{w_t},~D$), whereas the quantization is performed for $500$ dimensions. We observe that the achievable bounds obtained via \eqref{zador_lower_bound} and \eqref{zador_upper_bound} are quite tight (they have a gap of approximately $0.0014$ bits/dimension) whereas the gap between the lower bound \eqref{para_sol_eq.1} with the achievable upper bound \eqref{achievable_bound_ecdq} approximated by \eqref{zador_lower_bound} is $0.0097$ bits/dimension, and the one approximated by \eqref{zador_upper_bound} is approximately $0.011$ bits/dimension. Thus, compared to the first experiment where $p=24$, the gap between the bounds on the minimum achievable rate $R^{\op}_t$ is considerably decreased because we increased the number of dimensions in the system. Clearly, when the number of dimensions in the system increase, the gap between \eqref{para_sol_eq.1} and the high dimensional approximations of \eqref{achievable_bound_ecdq} will become arbitrary small. The two bounds will coincide as $p\longrightarrow\infty$, because then, the gap of coding noise from Gaussianity goes to zero (see, e.g., \cite{zador:1982}, \cite[Lemma 1]{zamir-feder:1996b}) and also because for $p\longrightarrow\infty$, \eqref{zador_lower_bound} is equal to \eqref{zador_upper_bound} (see, e.g.,\cite[equation (83)]{conway-sloane1999}).} 

\section{Application in NCSs}\label{sec:ncs}

\par In this section, {we demonstrate the sequential coding framework in the NCS setup of Fig. \ref{fig:scalar_ncs} by applying the results obtained in {\S}\ref{sec:examples}. }
\begin{figure}[htp]
\centering
\includegraphics[width=0.9\columnwidth]{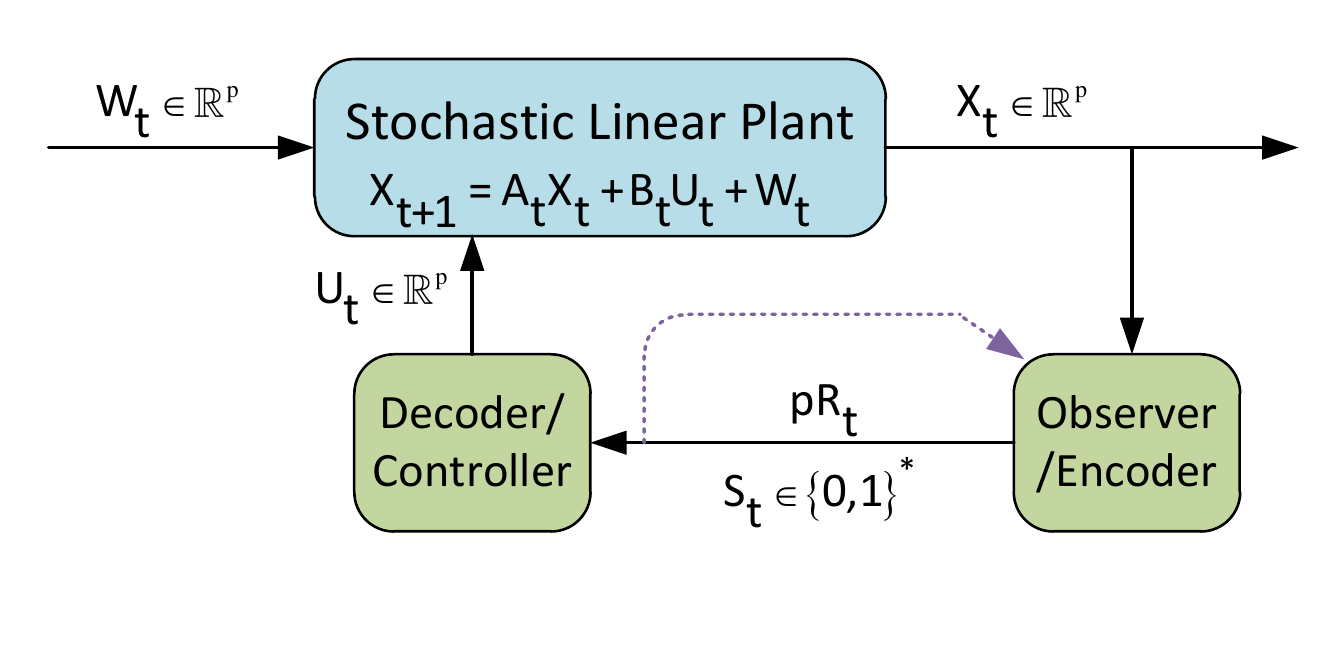}
\caption{Controlled system model.} 
\label{fig:scalar_ncs}
\end{figure}
 We first, describe each component of Fig. \ref{fig:scalar_ncs}.\\
\noindent{\bf Plant.} Consider $p$ parallel time-varying controlled Gauss-Markov processes as follows: 
\begin{align}
x_{t+1}(i)=\alpha_{t}x_{t}(i)+\beta_{t}u_t(i)+w_{t}(i),~i\in\mathbb{N}_1^p,~t\in\mathbb{N}_1^n,\label{example:controlled_process}
\end{align}
where $x_1(i)\equiv{x}_1$ is given with $x_1\sim{\cal N}(0;\sigma^2_{x_1})$,~$\forall{i}$; the non-random coefficients $(\alpha_t,\beta_t)\in\mathbb{R}$ are known to the system {with $(\alpha_t,\beta_t)\neq{0},~\forall{t}$}; {$\{u_t(i):~i\in\mathbb{N}_1^p\}$ is the controlled process with $u_t(i)\neq{u}_t(\ell),~$ for any $(i,{\ell})\in\mathbb{N}_1^p$}; $\{w_t(i)\equiv{w}_t:~i\in\mathbb{N}_1^p\}$ is an independent Gaussian noise process such that $w_t\sim{\cal N}(0;\sigma^2_{w_{t}})$, {$\sigma^2_{w_t}>0$}, independent of $x_1$,~$\forall{i}$. Again, similar to \S\ref{sec:examples}, \eqref{example:controlled_process} can be compactly written as follows
\begin{align}
X_{t+1}=A_{t}X_{t}+B_tU_t+W_{t},~X_1=\mbox{given},~t\in\mathbb{N}_1^n,\label{example:controlled_process1}
\end{align}
where {$A_t=\diag\left(\alpha_{t},\ldots,\alpha_t\right)\in\mathbb{R}^{p\times{p}}$, $B_t=\diag\left(\beta_{t},\ldots,\beta_t\right)\in\mathbb{R}^{p\times{p}}$,~$U_t\in\mathbb{R}^p$,~ $W_t\in\mathbb{R}^p\sim{\cal N}(0;\Sigma_{W_t})$, $\Sigma_{W_t}=\diag\left(\sigma^2_{w_t},\ldots,\sigma^2_{w_t}\right)\succ{0}$ is an independent Gaussian noise process independent of $X_1$}. Note that in this setup, the plant is fully observable for the observer that acts as an encoder but not for the controller due to the quantization noise (coding noise).\\
{\bf Observer/Encoder.} At the encoder the controlled process is collected into a frame $X_t\in\mathbb{R}^p$ from the plant and encoded as follows:
\begin{align}
S_t=f_t(X_{1,t},S_{1,t-1}),\label{example2:encoding}
\end{align}
where at $t=1$ we have $S_1=f_1(X_1)$, and $R_t=\frac{{\bf E}|S_t|}{p}$ is the rate at each time instant $t$ available for transmission via the noiseless channel. Note that in the design of Fig. \ref{fig:scalar_ncs}, the channel is noiseless, and the controller/decoder are deterministic mappings, thus, the observer/encoder implicitly has access to earlier control signals $U_{1,t-1}\in\mathbb{U}_{1,t-1}$.\\
{\bf Decoder/Controller.} The data packet $S_t$ is received by the controller using the following reconstructed sequence:
\begin{align}
U_t=g_t(S_{1,t}).~\label{example2:decoding}
\end{align}
According to \eqref{example2:decoding}, when the sequence $S_{1,t}$ is available at the decoder/controller, all past control signals $U_{1,t-1}$ are completely specified. \\
{\bf Quadratic cost.} The cost of control (per dimension) is defined as 
\begin{align}\label{eq:cost}
{\lqg}_{1,n}&\triangleq
\frac{1}{p}{\bf E}\left\{{\sum_{t = 1}^{n-1}\left(X_t\T\widetilde{Q}_tX_t  + U_t\T\widetilde{N}_tU_t\right) + X_n\T\widetilde{Q}_nX_n}\right\} ,
\end{align}
where $\widetilde{Q}_t=\diag\left(Q_t\ldots,Q_t\right)\succeq{0},~\widetilde{Q}_t\in\mathbb{R}^{p\times{p}}$ and $\widetilde{N}_t=\diag\left(N_t,\ldots,N_t\right)\succ{0},~\widetilde{N}_t\in\mathbb{R}^{p\times{p}}$, are designing parameters that penalize the state variables or the control signals.\\
{\bf Performance.} The performance of Fig. \ref{fig:scalar_ncs} (per dimension) can be cast to a finite-time horizon quantized \text{LQG} control problem subject to all communication constraints as follows:
\begin{align}
{\Gamma^{\IID,\op}_{\ssum}}(R)=\min_{\substack{(f_t,~g_t):~t=1,\ldots,n\\ \frac{1}{n}\sum_{t=1}^nR_t\leq{R}}}{\lqg}_{1,n}.\label{example2:problem_form}
\end{align}

\paragraph*{Iterative Encoder/Controller Design} In general, as \eqref{example2:problem_form} suggests, the optimal performance of the system in Fig. \ref{fig:scalar_ncs} is achieved only when
the encoder/controller pair is designed jointly. This is a quite challenging task especially when the channel is {noisy} because information structure {is} non-nested in such cases (for details see, e.g., \cite{yuksel-basar:2013}). There are examples, however, where the separation principle applies and the task comes much easier. More precisely, the so-called certainty equivalent controller remains optimal if the estimation errors are independent of previous control commands (i.e., dual effect is absent) \cite{bar-shalom:1974}. In our case, the optimal control strategy will be a certainty equivalence controller {\it if we assume a fixed and given sequence of encoders $\{f^*_t:~t\in\mathbb{N}_1^n\}$ {and the corresponding quantizer follows a predictive quantizer policy (similar to the $\dpcm$-based $\ecdq$ scheme proposed in \S\ref{subsec:upper_bound_mmse}), i.e., at each time instant it subtracts the effect of the previous control signals at the encoder and adds them at the decoder}} (see, e.g., \cite[Proposition 3]{bao-skoglund-johansson2011}, \cite{fu:2012}, \cite[\S{III}]{yuksel:2014}). Moreover, the separation principle will also be optimal if we consider an $\mmse$ estimate of the state (similar to what we have established in \S\ref{sec:examples}), and an encoder that minimizes a distortion for state estimation at the controller. The resulting separation principle is termed ``{\it weak separation principle}'' \cite{fu:2012} as it relies on the fixed and given quantization policies. This is different from the well-known full separation principle in the classical LQG stochastic control problem \cite{bertsekas:2005} where the problem separates naturally into a state estimator and a state feedback controller without any loss of optimality. 
{The previous analysis is described by a modified version of \eqref{example2:problem_form} as follows
\begin{align}
{\Gamma^{\IID,\op}_{\ssum}(R)\leq\Gamma^{\IID,\op,ws}_{\ssum}}=\min_{\substack{(f^*_t,~g_t):~t=1,\ldots,n\\ \frac{1}{n}\sum_{t=1}^nR_t\leq{R}}}{\lqg}_{1,n}.\label{example2:problem_form_weak_sep}
\end{align}
}
Next, we give the known solution of \eqref{example2:problem_form_weak_sep} in the form of a lemma that was first derived in \cite{tatikonda:2004,fu:2012} for the more general setup of correlated vector-valued controlled Gauss-Markov processes with linear quadratic cost. 
\begin{lemma}(Weak separation principle for Fig. \ref{fig:scalar_ncs})\label{lemma:scalar:separation}
The optimal controller that minimizes \eqref{example2:problem_form} is given by 
\begin{align}\label{eq:lem:scalar:separation:controller}
U_t = -L_t{\bf E}\left\{X_t|S_{1,t}\right\}, 
\end{align}
where ${\bf E}\left\{X_t|S_{1,t}\right\}$ are the fixed quantized state estimates obtained from the estimation problem in \S\ref{sec:examples}; $\widetilde{L}_t=\diag(L_t,\ldots,L_t)\in\mathbb{R}^p$ is the optimal $\lqg$ control (feedback) gain obtained as follows:
\begin{align}\label{eq:LQG:scalar:Kt}
\tilde{L}_t &=\left(B^2_t\widetilde{K}_{t+1}+\widetilde{N}_t\right)^{-1}B_t\widetilde{K}_{t+1}A_t,
\end{align}
and $\widetilde{K}_t=\diag(K_t,\ldots,K_t)\succeq{0}$ is obtained using the backward recursions:
\begin{align}\label{eq:lem:LQG:Lt}
\widetilde{K}_t &=A_t^2\left(\widetilde{K}_{t+1}-\widetilde{K}_{t+1}B^2_t(B_t^2\widetilde{K}_{t+1}+\widetilde{N}_t)^{-1}\widetilde{K}_{t+1}\right)+\widetilde{Q}_t,
\end{align}
with $\widetilde{K}_{n+1} = 0$.	Moreover, this controller achieves a minimum linear quadratic cost of 
\begin{align}
\begin{split}
{\Gamma^{\IID,\op,ws}_{\ssum}}=&\frac{1}{p}\sum_{t = 1}^n\Big\{\trace(\Sigma_{W_t}\widetilde{K}_t)+\trace(A_tB_t\widetilde{L}_t \widetilde{K}_{t+1}{\bf E}\{||X_t - Y_t||_2^2\})\Big\},\label{total_min_lqg}
\end{split}
\end{align}
where ${\bf E}\{||X_t - Y_t||_2^2\}$ is the $\mmse$ distortion obtained using any quantization (coding) in the control/estimation system. 
\end{lemma}
{
Before we prove our main theorem, we define the instantaneous cost of control as follows:
\begin{align}
\begin{split}
&{\lqg}^{\op}_t\triangleq\frac{1}{p}\Big\{\trace(\Sigma_{W_t}\widetilde{K}_t)+\trace(A_tB_t\widetilde{L}_t \widetilde{K}_{t+1}{\bf E}\{||X_t - Y_t||_2^2\})\Big\},~t\in\mathbb{N}_1^n.\end{split}\label{per_time_cost_of_control}
\end{align}
}
\par Next, we use Lemma \ref{lemma:scalar:separation} to derive a lower bound on \eqref{example2:problem_form_weak_sep}.   
{
\begin{theorem}({Lower bound on \eqref{example2:problem_form_weak_sep}})\label{theorem:fund_limt_lqg}
{For fixed coding policies, the minimum total-cost of control (per dimension) of \eqref{example2:problem_form_weak_sep}}, for any ``$n$'' and any $p$, however large, is ${\Gamma^{\IID,\op,ws}_{\ssum}}=\sum_{t=1}^n{\lqg}^{\op}_{t}$, with ${\lqg}^{\op}_{t}\geq{\lqg}^*_t$ such that
\begin{align}
{\lqg}_{t}^*=\sigma^2_{w_t}K_t + \alpha_t\beta_t L_t K_{t+1}D(R^*_t), \label{para_sol_cost_contr:eq.1}
\end{align}
where $D(R^*_t)$ is given by:
\begin{align}
D(R^*_t)&\triangleq\left\{ \begin{array}{ll} \frac{\sigma^2_{w_t}}{2^{2R^*_t}-\alpha^2_t},  & \mbox{$\forall{t}\in\mathbb{N}_1^{n-1}$} \\
2^{-2R^*_n}, &  \mbox{for $t=n$} \end{array} \right.,\label{min_dist:eq.1}
\end{align}
with the pair $(D(R^*_t),R^*_t)$ given by \eqref{para_sol_eq.1}-\eqref{values_of_xis}.
\end{theorem}
\begin{proof}
See Appendix \ref{proof:theorem:fund_limt_lqg}. 
\end{proof}
}
{In what follows, we include a technical remark related to the {\it lower bound} on the total cost-rate function of Theorem \ref{theorem:fund_limt_lqg}.}
{\begin{remark}(Technical remarks on Theorem \ref{theorem:fund_limt_lqg})\label{remark:technical_remarks_lower_bound}
 The expression of the lower bound in Theorem \ref{theorem:fund_limt_lqg}, can be reformulated for any $n$, and any $p$, to the equivalent expression of the total rate-cost function, denoted hereinafter by $\sum_{t=1}^nR({\lqg}_t^*)$, as follows
\begin{align}
R({\lqg}_t^*)=\frac{1}{2}\log_2\left(\alpha_t^2+\frac{\alpha_t\beta_tL_tK_{t+1}\sigma^2_{w_t}}{{\lqg}^*_t-\sigma^2_{w_t}K_t}\right), ~t\in\mathbb{N}_1^{n-1},\label{rate_cost_function}
\end{align} 
with $R({\lqg}_n^*)\equiv{R}_n^*$ as it is independent of ${\lqg}_n^*$. Interestingly, one can observe that by substituting in \eqref{rate_cost_function} the per-dimension version of \eqref{eq:LQG:scalar:Kt}  we obtain
\begin{align}
{R}({\lqg}_t^*)&=\frac{1}{2}\log_2\left(\alpha_t^2\left(1+\frac{\frac{\beta_t^2K^2_{t+1}\sigma^2_{w_t}}{\beta_t^2K^2_{t+1}+N_t}}{\lqg^*_t-\sigma^2_{w_t}K_t}\right)\right),\\
&=\frac{1}{2}\left[\log_2(\alpha_t^2)+\log_2\left(1+\frac{\frac{\beta_t^2K^2_{t+1}\sigma^2_{w_t}}{\beta_t^2K^2_{t+1}+N_t}}{\lqg^*_t-\sigma^2_{w_t}K_t}\right)\right].\label{rate_cost_alternative_def}
\end{align}
The bound in \eqref{rate_cost_alternative_def} extends the result of \cite[Equation (16)]{kostina:2019} from an asymptotically average total-rate cost function to the case of a total-rate cost function where at each instant of time  the rate-cost function is obtained using an allocation of $\lqg_t^*$ obtained due to the rate-allocation of the quantized state estimation problem of Theorem \ref{theorem:fundam_lim_mmse}. Additionally, the expression in \eqref{rate_cost_alternative_def} reveals an interesting observation regarding the absolute minimum data rates for mean square stability of the plant (per dimension), i.e., $\sup_t{\bf E}\{(x_t)^2\}<\infty$ (see, e.g., \cite[Eq. (25)]{nair:2004} for the definition) for a fixed finite time horizon. In particular, \eqref{rate_cost_alternative_def} suggests that for unstable time-varying plants with arbitrary disturbances modeled as in \eqref{example:controlled_process1}, and provided that at each time instant the cost of control (per dimension) is with communication constraints, i.e., ${\lqg}_t^*>\sigma^2_{{w}_t}K_t$ (the derivation without communication constraints is well known as the separation principle holds without a loss and ${\lqg}_t^*=\sigma^2_{{w}_t}K_t,~\forall{t}$ \cite{bertsekas:2005}), then, the minimum possible rates at each time instant $t$, namely, ${R}({\lqg}_t^*)$, cannot be lower than $\log_2|\alpha_t|$, when $|\alpha_t|>1$. This result extends known observations  for time-invariant plants (see e.g., \cite[Remark 1]{kostina:2019}) to parallel and (possibly unbounded) time-varying plants for any fixed finite time horizon. 
\end{remark}
}
Next, we use Theorem \ref{theorem:achievability} to find an upper bound on ${\Gamma^{\IID,\op,ws}_{\ssum}}$.
\begin{theorem}(Upper bound on \eqref{example2:problem_form_weak_sep})\label{theorem:achievability_lqg} Suppose that in the system of Fig. \ref{fig:scalar_ncs}, the fixed coding policies are obtained using the predictive coding scheme via sequential causal $\dpcm$-based  $\ecdq$ coding scheme with an $\mathbb{R}^p$-valued lattice quantizer described in Theorem \ref{theorem:achievability}. Then, ${\Gamma^{\IID,\op,ws}_{\ssum}}=\sum_{t=1}^n{\lqg}^{\op}_t$ for any $n$, and any $p$, with the instantaneous cost of control $\{\lqg_t:~t\in\mathbb{N}_1^{n-1}\}$ (per dimension) to be upper bounded as follows:
\begin{equation}
{\lqg}^{\op}_t\leq\sigma^2_{w_t}K_t+ \alpha_t\beta_t L_t K_{t+1}\frac{4^{\frac{1}{p}}(2\pi{e}G_p)\sigma^2_{w_t}}{2^{2R^{\op}_t}-4^{\frac{1}{p}}(2\pi{e}G_p)\alpha_t^2},
\label{achievable_bound_ecdq_lqg}
\end{equation} 
whereas, at $t=n$, ${\lqg}^{\op}_n=\sigma^2_{w_n}K_n$ and $R^{\op}_t$ is bounded above as in \eqref{achievable_bound_ecdq}.
\end{theorem}
\begin{IEEEproof}
See Appendix \ref{proof:theorem:achievability_lqg}.
\end{IEEEproof}

\begin{remark}(Comments on Theorem \ref{theorem:achievability_lqg})
For infinitely large spatial components, i.e., $p\longrightarrow\infty$, the upper bound in \eqref{achievable_bound_ecdq_lqg} approaches the lower bound in Theorem \ref{theorem:fund_limt_lqg} because $G_\infty\longrightarrow\frac{1}{2\pi{e}}$ (see e.g, \cite[Lemma 1]{zamir-feder:1996b}). {Moreover, one can easily obtain the equivalent inverse problem of the total rate-cost function for the upper bound in \eqref{achievable_bound_ecdq_lqg} similar to Remark \ref{remark:technical_remarks_lower_bound}}.
\end{remark}

{ Next, we note the main technical difference of both Theorems \ref{theorem:fund_limt_lqg}, \ref{theorem:achievability_lqg} compared to existing results in the literature. 
\begin{remark}(Connections to existing works)
{\bf (1)} Our bounds on $\lqg$ cost extend similar bounds derived in \cite[Theorems 7, 8]{khina:2018} to average total-rate constraints for any fixed finite time horizon. This constraint requires the use of a dynamic reverse-waterfilling optimization algorithm (derived in Theorem \ref{theorem:fundam_lim_mmse}) to optimally assign the rates at each instant of time for the whole fixed finite time horizon. In contrast, the fixed rate constraint (averaged across the time) assumed in \cite[Theorem 7, 8]{khina:2018} does not require a similar optimization technique because at each instant of time the transmit rate is the same. Another structural difference compared to \cite[Theorem 7, 8]{khina:2018} is that in our bound we decouple the dependency of $D_{t-1}$ at each time instant. \\
{\bf (2)} Our results also extend the steady-state bounds on $\lqg$ cost obtained in \cite{silva:2011,silva:2016,tanaka:2016} to cost-rate functions constrained by total-rates obtained for any fixed finite time horizon. By assumption, the rate constraint in those papers implies fixed (uniform) rates at each instant of time whereas our bounds require a rate allocation algorithm to assign optimally the rate at each time slot. 
\end{remark}
\subsection{Steady-state solution of Theorems \ref{theorem:fund_limt_lqg}, \ref{theorem:achievability_lqg}}\label{subsec:steady-state_solution_lqg}}

{In this subsection, we study the steady-state case of the bounds derived in Theorems  \ref{theorem:fund_limt_lqg}, \ref{theorem:achievability_lqg}. We start by making the following assumptions, i.e.,
\begin{itemize}
\item[(A1)] we restrict the controlled process \eqref{example:controlled_process1} to be time invariant, which means that $A_t\equiv{A}=\diag\left(\alpha,\ldots,\alpha\right)\in\mathbb{R}^{p\times{p}}$, $B_t\equiv{B}=\diag\left(\beta,\ldots,\beta\right)\in\mathbb{R}^{p\times{p}}$,~$W_t\in\mathbb{R}^p\sim{\cal N}(0;\Sigma_{W})$, $\Sigma_{W}=\diag\left(\sigma^2_{w},\ldots,\sigma^2_{w}\right)\succ{0},~\forall{t}$;
\item[(A2)] we restrict the design parameters that penalize the control cost \eqref{eq:cost} to also be time invariant, i.e., $\widetilde{Q}_t\equiv\diag(Q,\ldots,Q)$, $\widetilde{N}_t\equiv\diag(N, \ldots, N)$;
\item[(A3)] we fix $D_t\equiv{D}$, $\forall{t}$.
\end{itemize}
}
{
We denote the steady-state value of the total cost of control,  (per dimension) as follows:
\begin{align}
{\lqg}_{\infty}=\limsup_{n\longrightarrow\infty}\frac{1}{n}\sum_{t=1}^n{\lqg}_t.\label{steady_state_lqg}
\end{align}
\noindent{\bf Steady-state Performance.} The minimum achievable steady-state performance (per dimension) of the quantized $\lqg$ control problem of Fig. \ref{fig:scalar_ncs} under the weak separation principle can be cast to the following optimization problem:
\begin{align}
\Gamma^{\IID,\op,ws}_{\ssum,ss}=\min_{\substack{(f^*_t,~g_t):~t=1,\ldots,\infty\\
R_\infty\leq{R}}}{\lqg}_{\infty}.\label{performance_steady_state_lqg}
\end{align}
}
{In the next two corollaries, we prove the lower and upper bounds on \eqref{performance_steady_state_lqg}. These bounds follow from the assumptions (A1)-(A3) and Corollaries \ref{corollary:steady_state_rev_water}, \ref{corollary:achievability_ss}.}
{
\begin{corollary}(Lower bound on \eqref{performance_steady_state_lqg})\label{corollary:steady_state_lqg_lower}
The minimum achievable steady state performance of \eqref{performance_steady_state_lqg}, under the assumptions (A1)-(A3),  for any  $p$, is such that $\Gamma^{\IID,\op,ws}_{\ssum,ss}\geq{\lqg}^{*}_{\infty}$, where
\begin{align}
{\lqg}^{*}_{\infty}=\sigma^2_{w}K_\infty + \alpha\beta L_\infty K_\infty\frac{\sigma^2_{w}}{2^{2R_{\infty}^*}-\alpha^2},\label{steady_state_lqg_lower_bound}
\end{align}
where ${\lqg}^*_{\infty}\triangleq\lim_{n\longrightarrow\infty}\frac{1}{n}\sum_{t=1}^n{\lqg}_t^*$, with $L_\infty, K_\infty$ given by \eqref{eq:LQG:scalar:Kt_sted_s} and \eqref{ss_explicit_K}, respectively. 
\end{corollary}
\begin{IEEEproof}
The derivation follows from the assumptions (A1)-(A3).
In particular, 
\begin{align}
\frac{1}{n}\sum_{t=1}^n{\lqg}_t\stackrel{(i)}\geq&\frac{1}{n}\sum_{t=1}^n{\lqg}^{*}_t\nonumber\\
\stackrel{(ii)}=&\frac{1}{n}\sum_{t=1}^n\left(\sigma^2_{w}K_\infty + \alpha\beta L_\infty K_\infty{D}\right),\label{proof_ss_lqg_lower}
\end{align}
where $(i)$ follows from Theorem \ref{theorem:fund_limt_lqg};
$(ii)$ follows from the assumptions (A1)-(A3). In particular, by imposing the assumptions (A1), (A2), in Lemma \ref{lemma:scalar:separation} we obtain that the steady-steady optimal $\lqg$ control (feedback) gain (per dimension) becomes:
\begin{align}\label{eq:LQG:scalar:Kt_sted_s}
{L}_\infty=\frac{\alpha\beta{K_\infty}}{\beta^2{K}_{\infty}+{N}},
\end{align}
where $K_\infty$ is the positive solution of the quadratic equation:
\begin{align}\label{eq:lem:LQG:Lt_sted_s}
\beta^2{K}_\infty+\left((1-\alpha^2)N-\beta^2{Q}\right)K_\infty-QN =0,
\end{align}
given by the formula
\begin{align}
{K}_\infty=\frac{1}{2\beta^2}\left(\sqrt{\bar{f}^2+4\beta^2QN}-\bar{f}\right),\label{ss_explicit_K}
\end{align}
with $\bar{f}=(1-\alpha^2)N-\beta^2{Q}$. Finally by assumption (A3), we obtain from Corollary \ref{corollary:steady_state_rev_water} that $D\equiv{D}(R_t^*)=\frac{\sigma^2_{w}}{2^{2R_{\infty}^*}-\alpha^2}$, $\forall{t}$. The result follows once we let in \eqref{proof_ss_lqg_lower} $n\longrightarrow\infty$. This completes the derivation.
\end{IEEEproof}
}

{
\begin{corollary}(Upper bound on \eqref{performance_steady_state_lqg})\label{corollary:steady_state_lqg_upper}
The minimum achievable steady state performance of \eqref{performance_steady_state_lqg}, under the assumptions (A1)-(A3),  for any  $p$, is upper bounded as follows
\begin{align}
\Gamma^{\IID,\op,ws}_{\ssum,ss}\leq\sigma^2_{w}K_\infty + \alpha\beta L_\infty K_\infty\frac{4^{\frac{1}{p}}(2\pi{e}G_p)\sigma^2_{w}}{2^{2R^{\op}_{\infty}}-4^{\frac{1}{p}}(2\pi{e}G_p)\alpha^2},\label{steady_state_lqg_upper_bound}
\end{align}
where $R^{\op}_\infty$ is upper bounded by \eqref{achievable_bound_ecdq_steadu_state} and $K_\infty$, $L_\infty$ are given by \eqref{ss_explicit_K} and \eqref{eq:LQG:scalar:Kt_sted_s}, respectively. 
\end{corollary}
\begin{IEEEproof}
We omit the derivation because it is similar to the one obtained for the lower bound. In contrast to the lower bound, here we make use of Theorem \ref{theorem:achievability_lqg} and Corollary \ref{corollary:achievability_ss}. 
\end{IEEEproof}
Note that for Theorems \ref{theorem:fund_limt_lqg}, \ref{theorem:achievability_lqg} we can remark the following. 
\begin{remark}(Comments on Corollaries \ref{corollary:steady_state_lqg_lower}, \ref{corollary:steady_state_lqg_upper})\label{remark:conditions_stability}
The lower bound of Corollary \ref{corollary:steady_state_lqg_lower} (per dimension) is precisely the bound obtained by Tatikonda et al. in \cite[\S{V}]{tatikonda:2004} (see also \cite[\S{6}]{tatikonda:1998}) for scalar time-invariant Gauss-Markov processes. The upper bound of Corollary \ref{corollary:steady_state_lqg_upper} (per dimension) is similar to the upper bounds derived in \cite{silva:2011,silva:2016,khina:2018}. It is also similar to the upper bound obtained in \cite{tanaka:2016} albeit their space-filling term is obtained differently. 
\end{remark}
}
{
\section{Discussion and Open Questions}\label{sec:discussion}
}
{ In this section, we discuss certain open problems that can be solved based on this work and discuss certain observations that stem from our main results.
\subsection{Dynamic reverse-waterfilling algorithm for multivariate Gaussian processes}
Further to the technical observation raised in Remark \ref{remark:discussion:theorem:mmse}, {\bf (1)}, it seems that the simultaneous diagonalization of $(A_t, \Sigma_{W_t}, \Delta_t, \Lambda_t)$ by an orthogonal matrix  is sufficient in order to extend the derivation of a dynamic reverse-waterfilling algorithm to the more general case of multivariate time-varying Gauss-Markov processes. Our claim is further supported by the fact that for time-invariant multidimensional Gauss-Markov processes simultaneous diagonalization is shown to be sufficient for the derivation of a reverse-waterfilling algorithm in \cite[Corollary 1]{stavrou:2020}.
\subsection{Non-Gaussian processes}
Although not addressed in this paper, the non-asymptotic lower bounds derived in Theorems \ref{theorem:fundam_lim_mmse}, \ref{theorem:fund_limt_lqg} can be extended to linear models driven by independent non-Gaussian noise processes ${w_t}\sim(0;\sigma^2_{w_t})$ using entropy power inequalities \cite{kostina:2019}. 
\subsection{Packet drops with instantaneous ACK}
It would be interesting to extend our setup to the more practical scenario of communication links prone to packet drops. In such case one needs to take into account the various packet erasure models (e.g., $\IID$ or Markov models) to study their impact on the non-asymptotic bounds derived for the two application examples of this paper. Existing results for uniform (fixed) rate allocation are already studied in \cite{khina:2018}.
}

%
%
%
%
\section{Conclusion}\label{sec:conclusions}

\par {We revisited the sequential coding of correlated sources with independent spatial components to use it in the derivation of non-asymptotic, finite dimensional lower and upper bounds for two application examples in stochastic systems. Our application examples included a parallel time-varying quantized state-estimation problem subject to a total $\mse$ distortion constraint and a parallel time-varying quantized $\lqg$ closed-loop control system with linear quadratic cost. For the latter example, its lower bound revealed the minimum possible rates for mean square stability of the plant at each instant of time  when the system operates for a fixed finite time horizon. 
 }

\appendices

\section{Proof of Theorem \ref{theorem:fundam_lim_mmse}}\label{proof:theorem:fundam_lim_mmse}

Since the source is modeled as a time-varying first-order Gauss-Markov process, then from \eqref{eq:sumrate2}  we obtain:
\begin{align}
\begin{split}
&{\cal R}_{{\ssum}}^{\IID,\op,1}(D)\geq{\cal R}_{{\ssum}}^{\IID,1}(D),\\
&=\min_{\substack{\frac{1}{n}\frac{1}{p}\sum_{t=1}^n{\bf E}\left\{||X_t-Y_t||_2^2\right\}\leq{D},\\Y_1\leftrightarrow{X_{1}}\leftrightarrow{X}_{2,n},\\~Y_t\leftrightarrow(X_{t},Y_{1,t-1})\leftrightarrow(X_{1,t-1},{X}_{t+1,n})}}\frac{1}{p}\sum_{t=1}^nI(X_{t};Y_t|Y_{1,t-1})\end{split}.\label{exam:eq:sumrate1}
\end{align}
It is trivial to see that the $\rhs$ term in \eqref{exam:eq:sumrate1} corresponds precisely to the sequential or $\nrdf$ obtained for parallel Gauss-Markov processes with a total $\mse$ distortion constraint which is a simple generalization of the scalar-valued problem that has already been studied in \cite{stavrou:2018lcss}.  Therefore, using the analysis of \cite{stavrou:2018lcss} we can obtain:
\begin{align}
&{\cal R}_{\ssum}^{\IID,1}(D)\nonumber\\
&\stackrel{(a)}=\min_{\mbox{contraint in \eqref{exam:eq:sumrate1}}}\frac{1}{p}\sum_{t=1}^n\left\{h(X_t|Y_{1,t-1})-h(X_t|Y_{1,t})\right\}\nonumber\\
&\stackrel{(b)}=\frac{1}{p}\min_{\substack{\Delta_t\succeq{0},~t\in\mathbb{N}_1^n\\ \frac{1}{n}\frac{1}{p}\sum_{t=1}^n\trace{(\Delta_t)}\leq{D}}} \sum_{t=1}^n  \max\left[0,\frac{1}{2}\log_2\left(\frac{|\Lambda_t|}{|\Delta_t|}\right)\right],\nonumber\\
&=\min_{\substack{D_t\geq{0},~t\in\mathbb{N}_1^n\\\frac{1}{n}\sum_{t=1}^nD_t\leq{D}}} \sum_{t=1}^n  \max\left[0,\frac{1}{2}\log_2\left(\frac{\lambda_t}{D_t}\right)\right],\label{scalar:eq.solution}
\end{align}
{where $(a)$ follows by definition; $(b)$ follows from the fact that $h(X_t|Y_{1,t-1})=\frac{1}{2}\log_2(2\pi{e})^p|\Lambda_t|$ where $\Lambda_t=\diag\left(\lambda_t,\ldots\lambda_t\right)\in\mathbb{R}^{p\times{p}}$ with $\lambda_t=\alpha^2_{t-1}D_{t-1}+\sigma^2_{w_{t-1}}$, and that $h(X_t|Y_{1,t})=\frac{1}{2}\log_2(2\pi{e})^p|\Delta_t|$ where $\Delta_t=\diag\left(D_t,\ldots,D_t\right)\in\mathbb{R}^{p\times{p}}$ for $D \in [0, \infty)$. The optimization problem of  (\ref{scalar:eq.solution}) is already solved in \cite[Theorem 2]{stavrou:2018lcss} and is given by \eqref{para_sol_eq.1}-\eqref{values_of_xis}.}

\section{Proof of Theorem \ref{theorem:achievability}}\label{proof:theorem:achievability}

{In this proof we bound the rate performance of the $\dpcm$ scheme described in \S\ref{subsec:upper_bound_mmse} at each time instant for any fixed finite time $n$ using an $\ecdq$ scheme that utilizes the forward Gaussian test-channel realization that achieves the lower bound of Theorem \ref{theorem:fundam_lim_mmse}. In this scheme in fact we replace the quantization noise with an additive Gaussian noise with the same second moments.\footnote{See e.g., \cite{zamir-feder:1996} or \cite[Chapter 5]{zamir:2014} and the references therein.} First note that the Gaussian test-channel linear realization of the lower bound in Theorem \ref{theorem:fundam_lim_mmse} is known to be\cite{stavrou:2018lcss}
\begin{align}
Y_t=H_tX_t+(I_p-H_t)A_{t-1}Y_{t-1}+H^{\frac{1}{2}}V_t,~V_t\sim{\cal N}(0;\Delta_t),\label{realization}
\end{align} 
where $H_t\triangleq{I_p-\Delta_t\Lambda^{-1}_t}\succeq{0}$,~$\Delta_t\triangleq\diag(D_t,\ldots,D_t)\succ{0}$,~$\Lambda_t=\diag\left(\lambda_t,\ldots\lambda_t\right)\succ{0}$.\\
{\bf Pre/Post Filtered ECDQ with multiplicative factors for parallel sources.} \cite{zamir-feder:1996} First, we consider a $p-$dimensional lattice quantizer $Q_p$ \cite{conway-sloane1999} such that
\begin{align}
{\bf E}\{Z_tZ_t\T\}=\Sigma_{V^c_t}, ~\Sigma_{V_t^c}\succ{0},\nonumber
\end{align}
where $Z_t\in\mathbb{R}^p$ is a random dither vector generated both at the encoder and the decoder independent of the input signals $\widehat{X}_t$ and the previous realizations of the dither, uniformly distributed over the basic Voronoi cell of the $p-$dimensional lattice quantizer $Q_p$ such that $V_t^c\sim{Unif}(0;\Sigma_{V_t^c})$. At the {\it encoder} the lattice quantizer quantize $H_{t}^{\frac{1}{2}}\widehat{X}_t+Z_t$, that is, $Q_p(H_{t}^{\frac{1}{2}}\widehat{X}_t+Z_t)$ ,where $\widehat{X}_t$ is given by \eqref{dpcm_enc}. Then, the encoder applies entropy coding to the output of the quantizer and transmits the output of the entropy coder. At the {\it decoder} the coded bits are received and the output of the quantizer is reconstructed, i.e.,  $Q_p(H_{t}^{\frac{1}{2}}\widehat{X}_t+Z_t)$. Then, it generates an estimate by subtracting ${Z_t}$ from the quantizer's output and multiplies the result by $\Phi_t$ as follows:
\begin{align}
Y_t=\Phi_t(Q_p(H_{t}^{\frac{1}{2}}\widehat{X}_t+Z_t)-Z_t),\label{dec_ecdq}
\end{align}
where $\Phi_t=H_{t}^{\frac{1}{2}}$. The {\it coding rate at each instant of time} of the conditional entropy of the $\mse$ quantizer is given by\cite{zamir-feder:1996}
\begin{align}
H(Q_p|Z_t)&={I}(H^{\frac{1}{2}}\widehat{X}_t;H\widehat{X}_t+H^{\frac{1}{2}}V_t^c)\nonumber\\
&\stackrel{(a)}={I}(H^{\frac{1}{2}}\widehat{X}_t;H\widehat{X}_t+H^{\frac{1}{2}}V_t)+{\cal D}(V^c_t||V_t)\nonumber\\
&\qquad-{\cal D}(H\widehat{X}_t+H^{\frac{1}{2}}V^c_t||H\widehat{X}_t+H^{\frac{1}{2}}V_t)\nonumber\\
&\stackrel{(b)}\leq{I}(H^{\frac{1}{2}}\widehat{X}_t;H\widehat{X}_t+H^{\frac{1}{2}}V_t)+{\cal D}(V^c_t||V_t)\nonumber\\
&\stackrel{(c)}\leq{I}(H^{\frac{1}{2}}\widehat{X}_t;H\widehat{X}_t+H^{\frac{1}{2}}V_t)+\frac{p}{2}\log(2\pi{e}G_p)\nonumber\\
&\stackrel{(d)}={I}(X_t;Y_t|Y_{1,t-1})+\frac{p}{2}\log(2\pi{e}G_p)
\label{cond_entropy_quant}
\end{align}
where $V_t^c\in\mathbb{R}^p$ is the (uniform) coding noise in the $\ecdq$ scheme and $V_t$ is the corresponding Gaussian counterpart; $(a)$ follows because the two random vectors $V_t^c, V_t$ have the same second moments hence we can use the identity ${\cal D}(x||x')=h(x')-h(x)$; $(b)$  follows because ${\cal D}(H\widehat{X}_t+H^{\frac{1}{2}}V_t^c||H\widehat{X}_t+H^{\frac{1}{2}}V_t)\geq{0}$; $(c)$ follows because the divergence of the coding noise from Gaussianity is less than or equal to $\frac{p}{2}\log(2\pi{e}G_p)$ \cite{zamir-feder:1996b} where $G_p$ is the dimensionless normalized second moment of the lattice \cite[Definition 3.2.2]{zamir:2014}; $(d)$ follows from data processing properties, i.e., $I(X_t;Y_t|Y_{1,t-1})\stackrel{(\ast)}=I(X_t;Y_t|Y_{t-1})\stackrel{(\ast\ast)}=I(\widehat{X}_t;\widehat{Y}_t)\stackrel{(\ast\ast\ast)}={I}(H^{\frac{1}{2}}\widehat{X}_t;H\widehat{X}_t+H^{\frac{1}{2}}V_t)$  where $(\ast)$ follows from the realization of \eqref{realization}, $(\ast\ast)$ follows from the fact that $\widehat{X}_t$ and $\widehat{Y}_t$ (obtained by \eqref{dpcm_dec}) are independent of  $Y_{t-1}$, and $(\ast\ast\ast)$ follows from  \eqref{dpcm_enc}, \eqref{realization} and the fact that $H$ is an invertible operation. Since we assume joint (memoryless) entropy coding with lattice quantizers, then, the total coding rate per dimension is obtained as follows\cite[Chapter 5.4]{cover-thomas2006}
\begin{align}
\sum_{t=1}^n\frac{{\bf E}|S_t|}{p}&\leq\frac{1}{p}\sum_{t=1}^n\left({H}(Q_p|Z_t)+1\right)\nonumber\\
&\stackrel{(e)}\leq\frac{1}{p}\sum_{t=1}^n{I}({X}_t;{Y}_t|Y_{1,t-1})+\frac{n}{2}\log(2\pi{e}G_p)+\frac{n}{p}\nonumber\\
&\stackrel{(f)}=\frac{1}{2p}\sum_{t=1}^n\log_2\frac{|\Lambda_t|}{|\Delta_t|}+\frac{n}{2}\log(2\pi{e}G_p)+\frac{n}{p},\label{ecdq_ineq}
\end{align}
where $(e)$ follows from \eqref{cond_entropy_quant}; $(f)$ follows from the derivation of Theorem \ref{theorem:fundam_lim_mmse}. The derivation is complete once we minimize both sides of inequality in \eqref{ecdq_ineq} with the appropriate constraint sets. }

\section{Proof of Theorem \ref{theorem:fund_limt_lqg}}\label{proof:theorem:fund_limt_lqg}
{
Note that from \eqref{total_min_lqg} we obtain
\begin{align}
\begin{split}
&{\Gamma^{\IID,\op,ws}}=\sum_{t=1}^n{\lqg}^{\op}_t\\
&=\frac{1}{p}\sum_{t=1}^n\Big\{\trace(\Sigma_{W_t}\widetilde{K}_t)\\
&+\trace(A_tB_t\widetilde{L}_t \widetilde{K}_{t+1}{\bf E}\{||X_t - Y_t||_2^2\})\Big\}\\
&\stackrel{(a)}\geq\frac{1}{p}\sum_{t=1}^n\Big\{\trace(\Sigma_{W_t}\widetilde{K}_t)\\
&+\trace(A_tB_t\widetilde{L}_t \widetilde{K}_{t+1}{\bf E}\{||X_t - {\bf E}\{X_t|S_{1,t}\}||_2^2\})\Big\}\\
&\stackrel{(b)}\geq\frac{1}{p}\sum_{t=1}^n\Bigg\{\trace(\Sigma_{W_t}\widetilde{K}_t)+\trace\Bigg(A_tB_t\widetilde{L}_t \widetilde{K}_{t+1}\\
&{\bf E}_{\bar{S}_{1,t-1}}\left\{\frac{1}{2\pi{e}}2^{\frac{2}{p}h(X_t|S_{1,t-1}=\bar{S}_{1,t-1})}\right\}2^{-2R^*_t}\Bigg)\Bigg\}\\
&\stackrel{(c)}\geq\frac{1}{p}\sum_{t=1}^n\Bigg\{\trace(\Sigma_{W_t}\widetilde{K}_t)+\trace\Bigg(A_tB_t\widetilde{L}_t \widetilde{K}_{t+1}\\
&\left\{\frac{1}{2\pi{e}}2^{\frac{2}{p}h(X_t|S_{1,t-1})}2^{-2R^*_t}\right\}\Bigg)\Bigg\},\\
&\stackrel{(d)}\geq\sum_{t=1}^n\Big\{\sigma^2_{w_t}K_t+\alpha_t\beta_tL_tK_{t+1}D(R_t^*)\Big\}\triangleq\sum_{t=1}^n{\lqg}^*_t,
\end{split}\label{lower_bound_lqg}
\end{align}
where $(a)$ follows from the fact that  $Y_t$ is $\mathbb{S}_{1,t}-$measurable and the $\mmse$ is obtained for $Y_t={\bf E}\{X_t|S_{1,t}\}$; $(b)$ follows from the fact that ${\bf E}\{||X_t - {\bf E}\{X_t|S_{1,t}\}||_2^2\})={\bf E}_{\bar{S}_{1,t-1}}\left\{{\bf E}\{||X_t - {\bf E}\{X_t|S_{1,t}\}||_2^2|S_{1,t-1}=\bar{S}_{1,t-1}\}\right\}$, where ${\bf E}_{\bar{S}_{1,t}}\{\cdot\}$ is the expectation with respect to some vector $\bar{S}_{1,t-1}$ that is distributed similarly to $S_{1,t-1}$, also from the $\mse$ inequality in \cite[Theorem 17.3.2]{cover-thomas2006} and finally from the fact that $R^*_t\geq{0}$, where $R^*_t=\frac{1}{p}\left\{h^*(X_t|Y_{1,t-1})-h^*(X_t|Y_{1,t})\right\}$ (see the derivation of Theorem \ref{theorem:fundam_lim_mmse}, ({\bf 1})) with $h^*(X_t|Y_{1,t-1})$, $h^*(X_t|Y_{1,t})$ being the minimized values in \eqref{scalar:eq.solution}; $(c)$ follows from Jensen's inequality \cite[Theorem 2.6.2]{cover-thomas2006}, i.e., ${\bf E}_{\bar{S}_{1,t-1}}\left\{2^{\frac{2}{p}h(X_t|S_{1,t-1}=\bar{S}_{1,t-1})}\right\}\geq{2}^{\frac{2}{p}h(X_t|S_{1,t-1})}$; $(d)$ follows from the fact that $\{h(X_t|S_{1,t-1})=h(A_{t-1}X_{t-1}+B_{t-1}U_{t-1}+W_{t-1}|S_{1,t-1}):~t\in\mathbb{N}_2^n\}$ is completely specified from the independent Gaussian noise process $\{W_{t-1}:~t\in\mathbb{N}_2^n\}$ because $\{U_{t-1}=g_{t}(S_{1,t-1}):~t\in\mathbb{N}_2^n\}$ (see \eqref{example2:decoding}) are constants conditioned on $S_{1,t-1}$. Therefore, $h(X_t|S_{1,t-1})$ is conditionally Gaussian thus equivalent to $h(X_t|Y_{1,t-1})$. This further means that $\frac{1}{2\pi{e}}{2}^{\frac{2}{p}h(X_t|Y_{1,t-1})}2^{-2R^*_t}\geq\frac{1}{2\pi{e}}{2}^{\frac{2}{p}h^*(X_t|Y_{1,t-1})}2^{-2R^*_t}\stackrel{(\star)}=\frac{1}{2\pi{e}}{2}^{\frac{1}{p}\log_2(2\pi{e})^p|\Delta_t^*|}\stackrel{(\star\star)}=\min\{D_t\}\equiv{D}(R_t^*)$, where $(\star)$ follows because $h^*(X_t|Y_{1,t})=\frac{1}{2}\log_2(2\pi{e})^p|\Delta_t^*|$ and $(\star\star)$ follows because $\Delta_t^*=\diag(\min\{D_t\},\ldots,\min\{D_t\})$.}  
\par {It remains to find $D(R_t^*)$ at each time instant in \eqref{lower_bound_lqg}. To do so, we reformulate the solution of the dynamic reverse-waterfilling solution in \eqref{para_sol_eq.1} as follows:
\begin{align}
&nR\equiv{\cal R}_{\ssum}^{\IID,1}=\sum_{t=1}^n{R^*_t}\equiv\frac{1}{2}\sum_{t=1}^n \log_2\left(\frac{\lambda_t}{D_t}\right)\nonumber\\
&=\frac{1}{2}\left\{\underbrace{\cancelto{0}{\log_2(\lambda_1)}}_{\text{initial~step}}+\sum_{t=1}^{n-1}\log_2\left(\alpha^2_t+\frac{\sigma^2_{w_t}}{D_t}\right)-\underbrace{\log_2 D_n}_{\text{final~step}}\right\}.\label{reform_object_func}
\end{align} 
From \eqref{reform_object_func} we observe that at each time instant, the rate $R^*_t$ is a function of only one distortion $D_t$ since we have now decoupled the correlation with $D_{t-1}$. Moreover, we can assume without loss of generality, that the initial step is zero because it is independent of $D_0$. Thus, from \eqref{reform_object_func}, we can find at each time instant a $D_t\in(0,\infty)$ such that the rate is $R_t^*\in[{0},\infty)$. Since the rate distortion problem is equivalent to the distortion rate problem (see, e.g., \cite[Chapter 10]{cover-thomas2006}) we can immediately compute the total-distortion rate function, denoted by $D^{\IID,1}_{\ssum}(R)$, as follows: 
\begin{align}
D^{\IID,1}_{\ssum}(R)\triangleq\sum_{t=1}^nD(R_t^*)=\sum_{t=1}^{n-1} \frac{\sigma^2_{w_t}}{2^{2R^*_t}-\alpha_t^2}+2^{-2R^*_n}.\label{scalar:dist_rate1}
\end{align}
Substituting $D(R_t^*)$ at each time instant in \eqref{lower_bound_lqg} the result follows.}\\
This completes the proof. 

\section{Proof of Theorem \ref{theorem:achievability_lqg}}\label{proof:theorem:achievability_lqg}

{Note that from Lemma \ref{lemma:scalar:separation}, \eqref{total_min_lqg}, we obtain:
\begin{align}
\begin{split}
&{\Gamma^{\IID,\op,ws}_{\ssum}}=\frac{1}{p}\sum_{t = 1}^n\Big\{\trace(\Sigma_{W_t}\widetilde{K}_t)\\
&+\trace(A_tB_t\widetilde{L}_t \widetilde{K}_{t+1}{\bf E}\{||X_t - Y_t||_2^2\})\Big\}\\
&=\frac{1}{p}\sum_{t = 1}^n\Big\{\trace(\Sigma_{W_t}\widetilde{K}_t)\\
&+\trace(A_tB_t\widetilde{L}_t \widetilde{K}_{t+1}D(R^{\op}_t))\Big\}\\
&\stackrel{(a)}\leq\sum_{t=1}^{n-1}\left\{\sigma^2_{w_t}K_t+ \alpha_t\beta_t L_t K_{t+1}\frac{4^{\frac{1}{p}}(2\pi{e}G_p)\sigma^2_{w_t}}{2^{2R^{\op}_t}-4^{\frac{1}{p}}(2\pi{e}G_p)\alpha_t^2}\right\}\\
&+\sigma^2_{w_n}K_n,
\end{split}\label{proof:upper_bound_lqg_cost}
\end{align}
where $(a)$ {is obtained in two steps. As a first step, expand the inequality obtained in Theorem \ref{theorem:achievability}, \eqref{achievable_bound_ecdq} for the time horizon $n$ as follows
\begin{align}
R^{\op}_1+\ldots+R^{\op}_n\leq{R}_1^*+\ldots+R_n^*+c
\end{align}  
where $c=\frac{n}{p}\log_2(2\pi{e}G_p)+\frac{n}{p}$. As a second step, we reformulate $\{R_t^*:~t\in\mathbb{N}_1^n\}$ similar to \eqref{reform_object_func} (in the derivation of Theorem \ref{theorem:fund_limt_lqg}) so that we decouple the dependence  on $D_{t-1}$ at each time step.} Finally, for each $R^{\op}_t,~t=1,2\ldots,n,$ we solve the resulting inequality  with respect to $D(R_t^{\op})$ which gives
\begin{align}
\begin{split}
&D(R^{\op}_t)\leq\frac{4^{\frac{1}{p}}(2\pi{e}G_p)\sigma^2_{w_{t}}}{2^{2R^{\op}_{t}}-4^{\frac{1}{p}}(2\pi{e}G_p)\alpha_{t}^2},~t\in\mathbb{N}_1^{n-1}.
\end{split}\label{achievable_distortions}
\end{align}
Observe that the last step $t=n$ is not needed because in \eqref{total_min_lqg} we have ${K}_{n+1}=0$.
This completes the proof.}

{
\section*{Acknowledgement}
}
{The authors wish to thank the Associate Editor and the anonymous reviewers for their valuable comments and suggestions. They are also indebted to Prof. T. Charalambous for reading the paper and proposing the idea of bisection method for Algorithm \ref{algo1}. They are also grateful to Prof. J. {\O}stergaard for fruitful discussions on technical issues of the paper.}


\bibliographystyle{IEEEtran}

\bibliography{references}

\end{document}